\documentclass[12pt]{amsart}
\usepackage{amsfonts}
\usepackage{amsmath}
\usepackage{amssymb}

\usepackage{amscd}
\usepackage{mathrsfs}

\usepackage{amsthm}
\usepackage[latin1]{inputenc}
\usepackage{mathtools}
\usepackage{braket}
\usepackage[colorlinks=true, allcolors=magenta]{hyperref}
\usepackage{graphicx}%
\usepackage[capitalise]{cleveref}

\providecommand{\U}[1]{\protect\rule{.1in}{.1in}}

\newtheorem{theorem}{Theorem}
\newtheorem*{theorem*}{Theorem}

\newtheorem{conjecture}[theorem]{Conjecture}
\newtheorem{corollary}[theorem]{Corollary}

\newtheorem{definition}[theorem]{Definition}

\newtheorem{lemma}[theorem]{Lemma}
\newtheorem{notation}[theorem]{Notation}
\newtheorem{problem}[theorem]{Problem}

\newtheorem{remark}[theorem]{Remark}

\def\RE{\mathbb{R}}
\def\CO{\mathbb{C}}
\def\IN{\mathbb{N}}
\def\H{\mathcal H}
\def\D{\mathcal D}
\def\A{\mathcal A}
\def\G{\mathcal G}
\def\B{\mathcal B}
\def\N{\mathcal N}
\def\E{\mathrm{Extr}}

\def\F{\mathcal F}

\def\>{\rangle}
\def\<{\langle}
\def\K{\mathbb K}
\def\wt{\widetilde}
\def\ol{\overline}

\usepackage[margin=2.5cm]{geometry}
\begin{document}

\title[Wigner positive states]{Characterising the sets of quantum states with non-negative Wigner function}

\author[Cerf]{Nicolas J. Cerf}
\author[Chabaud]{Ulysse Chabaud}
\author[Davis]{Jack Davis}
\author[Dias]{Nuno C. Dias}
\author[Prata]{Jo\~ao N. Prata}
\author[Van Herstraeten]{Zacharie Van Herstraeten$^*$}
\thanks{$*$Authors are listed alphabetically.}

\date{}

\maketitle


\begin{abstract}
For Hilbert spaces $\H\subseteq L^2(\RE)$ we consider the convex sets $\D_+(\H)$ of Wigner-positive states (WPS), i.e.~density matrices over $\H$ with non-negative Wigner function. We investigate the topological structure of these sets, namely concerning closure, compactness, interior and boundary (in a relative topology induced by the trace norm). We also study their geometric structure and construct minimal sets of states that generate $\D_+(\H)$ through convex combinations.  If $\H$ is finite-dimensional, the existence of such sets follows from a central result in convex analysis, namely the Krein--Milman theorem. In the infinite-dimensional case $\H=L^2(\RE)$ this is not so, due to lack of compactness of the set $\D_+(\H)$. Nevertheless, we prove that a Krein--Milman theorem holds in this case, which allows us to extend most of the results concerning the sets of generators to the infinite-dimensional setting. Finally, we study the relation between the finite and infinite-dimensional sets of WPS, and prove that the former provide a hierarchy of closed subsets, which are also proper faces of the latter. These results provide a basis for an operational characterisation of the extreme points of the sets of WPS, which we undertake in a companion paper \cite{physicspaper}. Our work offers a unified perspective on the topological and geometric properties of the sets of WPS in finite and infinite dimensions, along with explicit constructions of minimal sets of generators. 
\end{abstract}


\tableofcontents


\section{Introduction}
\label{sec:intro}

Wigner functions are widely used in quantum mechanics, quantum optics, and quantum information \cite{GossonBook1,Lions,Wigner}. They are phase-space quasi-probability distributions providing a unified formulation of both pure and mixed quantum states. Wigner functions share many features with classical probability distributions, with the notable exception that they may assume negative values. 

This property is often interpreted as a signature of their quantum nature, signalling non-classicality \cite{Booth}, and can also be understood as a resource in the context of quantum computing \cite{Albarelli}. 
On the other hand, quantum states with non-negative Wigner functions (denoted Wigner-positive states (WPS) hereafter) are particularly interesting: they can be thought of as quantum states with a `classical' description and they lead to efficient classical algorithms for simulating quantum computations \cite{Mari}.

For pure quantum states, Hudson's theorem provides a simple characterisation of WPS: a Wigner function $W\psi(z)$ is non-negative if and only if it corresponds to a generalised Gaussian \cite{Hudson,Soto}. 
In contrast, for mixed states the identification and construction of all WPS remains a challenging problem, with little progress since Hudson's original publication more than fifty years ago. For various perspectives on this problem, as well as a few notable partial results, the reader may refer to \cite{Bondia,Werner,Dias1,Sudarshan,Cerf,Narcowich1,Narcowich2,van1}. 


\subsection{Overview}

In this paper, we take a novel approach to the 
\emph{Wigner positivity problem}. Our main objective is to investigate the topological and geometric structure of entire sets of WPS, rather than focusing on the properties of individual states. However, since the two problems are in fact related, as a by-product of our analysis we will also obtain explicit results concerning the characterisation and the construction of WPS. Our approach largely relies on methods and tools from convex analysis and Banach space theory.

We are mainly interested in the case where the Hilbert space of the system is $\H=L^2(\RE)$. However, most of the results for this case are, to some extend, related to similar properties of the finite-dimensional systems with Hilbert spaces $\H\subseteq \H^N=$ span$\{|n\>,\, {0\le n\le N} \}$, where $|n\>$ are the usual Fock states. Motivated by this, we will provide a unified approach to both the finite and infinite-dimensional cases, emphasising the connections between the two.  

Let $\D(\H)$ be the set of density matrices over $\H$, and $\D_+(\H) \subset \D(\H)$ denote the convex set of WPS. When $\H=\H^N$ we write $\D_+^N$; and if $\H=L^2(\RE)$ we simply write $\D_+$. We want to address the following two general problems:

\begin{problem}[Topological structure]\label{prob:topology}
What are the topological properties of the sets $\D_+(\H)$? Can one identify the interior and boundary points of $\D_+(\H)$ using simple properties of the corresponding Wigner functions? Are there any interesting dense subsets of $\D_+(\H)$? \\
\smallskip
\emph{(The topology considered here is the relative topology induced by the trace norm in the affine hull of $\D(\H)$).} 
\end{problem}

\begin{problem}[Geometric structure]\label{prob:geometry}
Is it possible to construct every state in $\D_+(\H)$ as a convex combination of a small set of \emph{generator} states? In both the finite and infinite-dimensional cases, what are the minimal generating sets, and can the generators be identified and constructed explicitly? Can all  states in $\D(\H)$ be written as affine combinations of the same set of generators? 
\end{problem}

We solve \cref{prob:topology} for all cases $\H \subseteq \H^N$ and $\H=L^2(\RE)$ in \cref{sec:topo}. In particular, in the infinite-dimensional setting we show that $\D_+$ is closed but non-compact, and has empty interior. On the other hand, in the finite-dimensional case, we prove that $\D_+(\mathcal H)$ is closed and compact. Moreover, when $\mathcal H\subseteq\mathcal H^N$, we provide a complete characterisation of the interior and boundary of the sets of WPS using a recent result about the properties of the sets of zeros of Wigner functions \cite{Abreu}.

\medskip

Building on these results, our approach to \cref{prob:geometry}, detailed in \cref{sec:geom}, further relies on the Krein--Milman theorem \cite{Krein,Simon2011}:

\begin{theorem}[Krein--Milman] \label{Krein--Milman}
Let $S$ be a compact and convex subset of a Hausdorff, locally convex topological vector space $V$. Then
\begin{equation}\label{KM1}
S = \overline{\mbox{\rm conv}}\bigl(\E(S)\bigr),
\end{equation}
where $\E(S)$ is the set of extreme points of $S$ (see \cref{extreme} below) and $
\overline{\mbox{\rm conv}}\bigl(\E(S)\bigr)$ denotes the smallest closed convex set containing $\E(S)$.

\end{theorem}

Since every extreme point is necessarily a boundary point in a convex set, one also has
$$
S = \overline{\mbox{\rm conv}}\bigl(\partial S\bigr).
$$

In the finite-dimensional case, where $\D_+(\H)$ is compact, these results indicate that the extreme and boundary points provide two natural sets of generators of $\D_+(\H)$. Using a recent result \cite{Abreu}, we are able to characterise these states in terms of simple properties of the nodal sets of their Wigner functions, and then present a method for constructing these states explicitly. This yields a solution to \cref{prob:geometry} in the finite-dimensional case.

The extension of these results to the infinite-dimensional case $\H=L^2(\RE)$ is non-trivial. In this setting, the set of WPS is non-compact and has empty relative interior in the trace norm topology; hence the Krein--Milman theorem does not apply. However, while the conditions in the Krein--Milman theorem are sufficient for the relation (\ref{KM1}) to hold, they are not necessary for such (or other) decomposition to exist. Therefore, we aim to identify a suitable subset of generators $\G\subset \D_+ $ that can play the role of ``pure states" for the WPS in the infinite-dimensional case:
$$
\D_+ = \overline{\mbox{\rm conv}}\{\G\}.
$$
Interestingly, by equipping the set of trace-class operators with the weak* topology, we are able to prove that our main results for \cref{prob:geometry} in the finite-dimensional case extend to the infinite-dimensional setting. In particular, we show that a Krein--Milman-type relation holds for the trace norm with $\G=\E(\D_+)$ or, more generally, $\G$ a dense subset of $\E(\D_+)$. A related result shows that the set of WPS with nonempty nodal set is dense in $\D_+$. 

We further investigate the relation between the finite-dimensional and infinite-dimensional cases and show, in particular, that the sets $\D_+^N$ form a hierarchy of closed subsets which are also {\it proper faces} of $\D_+$. This implies that all extreme points of  $\D_+^N$ are also extreme points of $\D_+$. 

\medskip


\subsection{Summary of results}

We begin by introducing some key notation and by defining the primary sets studied in this paper, after which we summarise the main results. For more details on notation, as well as general definitions and background results, the reader should refer to the sections \ref{Notation} and \ref{Preliminaries}.

For each $N\in \IN_0$, we define the finite-dimensional Hilbert space: 
$$
\H^N=\mathrm{span}\{|n\rangle,\, 0\le n\le N\} \subset L^2(\RE).
$$ 
where the Fock states $|n\rangle$, ${n\in \IN_0}$ are the eigenstates of the quantum harmonic oscillator Hamiltonian. In the position representation they  are given by the  Hermite functions:  
\begin{equation}\label{HermiteF}
h_n(x) =(2^n n!)^{-1/2} \pi^{-1/4} H_n(x)e^{-\tfrac{1}{2}x^2}, \quad n \in \IN_0,
\end{equation} 
where $H_n(x)$ is the Hermite polynomial of degree $n$ (we use units where $\hbar=1$).
These functions form an orthonormal basis of $L^2(\RE)$. For notational convenience, we identify $|n\>$ with the function $h_n$, and treat them interchangeably.

Unless otherwise stated, we assume the Hilbert space of the system to be either $\H\subseteq\H^N$ or $\H=L^2(\RE)$.

Quantum states (pure or mixed) are represented by density matrices, i.e.~self-adjoint, positive, trace-class operators $\rho:\H\to \H$ with unit trace. Pure states correspond to rank-one projectors. 

For a density matrix $\rho$ let $W\rho$ denote its Wigner function. In the sequel we will repeatedly use the two following sets:
\medskip

\begin{itemize}

\item The {\it nodal set} of $W\rho$ (and, by abuse of notation, of $\rho$) is defined as:
$$
\N(\rho)\equiv \N(W\rho)\coloneqq\{z\in \RE^2: W\rho(z)=0\}.
$$ 
\item The {\it negative set} of $\rho$ is:
 $$
 \Omega_-(\rho)\equiv \Omega_-(W\rho)\coloneqq\left\{z\in \RE^2: W\rho(z) \le 0 \right\}.
 $$
\end{itemize}


We also define the following sets of operators:
\medskip
\begin{itemize}

    \item \textbf{$\B_1(\H)$} is the space of trace-class operators. This is a Banach space with the (trace) norm $\|\rho\|_1={\rm tr}|\rho|$. Its elements are compact operators $\rho:\H\to \H$ satisfying $\|\rho\|_1 <\infty$.
    \medskip
    \item \textbf{$\A(\H)$} is the affine space of self-adjoint, trace-class operators with unit trace,
    $$
    \A(\H)=\left\{\rho \in \B_1(\H):\, \rho=\rho^\dagger,\, {\rm tr} \, \rho=1 \right\}.
    $$
    \item \textbf{$\D(\H)$} is the set of density matrices,
    $$
    \D(\H)=\left\{\rho \in \A(\H):\, \rho \ge 0 \right\}.
    $$
    \item \textbf{$\D_+(\H)$} is the set of Wigner-positive states (WPS),
    $$
    \D_+(\H)= \left\{\rho \in \D(\H):\, W\rho(z) \ge 0 \text{ for all } z\in \RE^2 \right\}.
    $$
    \item \textbf{$\D_{0+}(\H)$} is the subset of WPS with nonempty nodal set,
    $$
    \D_{0+}(\H)= \left\{\rho \in \D_+(\H):\, \N(\rho) \neq \emptyset \right\}.
    $$
    \item If $\H=L^2(\RE)$ we write the previous sets as $\B_1\, ,\A \, ,\D\, ,\D_+$ and $\D_{0+}$, respectively. 
    \medskip
    
    \item If $\H=\H^N$ we write them as $\B_1^N \, , \A^N \, , \D^N \, , \D_+^N $ and $\D_{0+}^N$, respectively. 
\end{itemize}
\medskip
These sets satisfy the chain of inclusions:
$$
\D_{0+}(\H)\subset \D_+(\H) \subset \D(\H) \subset \A(\H) \subset \B_1(\H).
$$

In the simplest case $\H=\H^1=$ span$\{|0\rangle ,|1\rangle
\}$, we obtain the nice geometrical structure displayed in Figure 1, and which immediately suggests several of the questions posed in \cref{prob:topology} and \cref{prob:geometry} regarding the topological and the geometric properties of $\D_+(\H)$. It also naturally prompts one to ask whether a similar structure appears when $\H\subseteq\H^N$ with $N>2$, or when $\H=L^2(\RE)$.

\begin{figure}[ht]
    \centering
    \includegraphics[width=7cm,height=7cm]{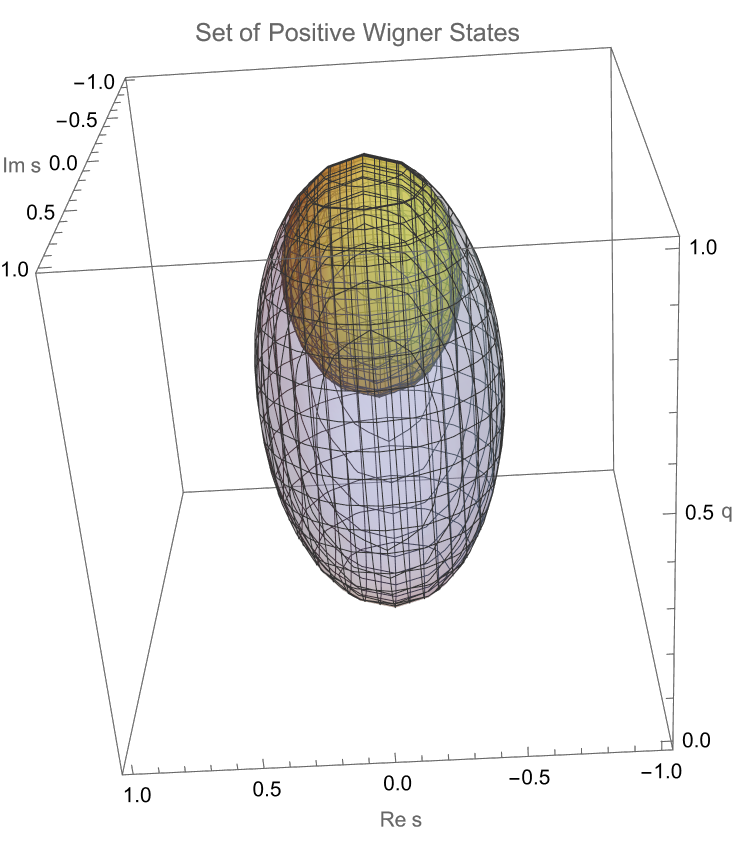} 
    \caption{\small The ambient space is the affine space $\A^1$. Each point $(q,{\rm Re}\, s,{\rm Im}\, s)$ represents a hermitian, trace-class operator $\rho =q|0\>\<0| +(1-q) |1\>\<1|+s|0\>\<1| + \ol{s} |1\>\<0|\in \A^1$. The outer ellipsoid represents the set $\D^1$, while the inner ellipsoid corresponds to $\D^1_+$. Moreover, the north pole is the state $|0\>\<0|$ and the south pole is $|1\>\<1|$. The Euclidean distance coincides with the trace distance.}
    \label{fig:example}
\end{figure}

\medskip

We address these questions in the paper. Here is a summary of the main results:

\medskip

\textbf{1. Finite-dimensional case ($\H\subseteq \H^N$):}

\begin{itemize}

\item If nonempty, the sets $\D_+(\H)$ are closed and compact (\cref{Compact_finite_dim}).

\medskip

\item A density matrix belongs to the interior of $\D_+(\H)$ (in the relative topology induced by the trace norm in $\A(\H)$) if and only if it is of maximal rank and has empty nodal set (\cref{ipos}). A WPS belongs to the (relative) boundary of $\D_+(\H)$ if and only if it has a nonempty nodal set or is a (relative) boundary point of $\D(\H)$ (\cref{boundary}).

\medskip

\item  If $|0\> \perp \H$ then $\D_+(\H)=\emptyset $. 
If $|0\> \in  \H$ then $\D_+(\H)$ has nonempty relative interior with respect to $\A(\H)$ (\cref{intWN}).

\medskip

\item With the (possible) exception of the Gaussian state $\rho_0 =|0\>\<0|$, all extreme points of $\D_+(\H)$ have nonempty nodal sets (\cref{Extreme_points}). If dim $\H \ge 2$, then $\D_+(\H)= \ol{\rm conv}(\D_{0+}^N)$ (\cref{GeneratorsVWPS}).

\medskip

\item For every closed set $S\subset \RE^2$, the set of density matrices $\F(S)=\{\rho \in \D_+(\H) \, |\, S \subseteq \N(\rho)  \}$, if nonempty, is a proper face of $\D_+(\H)$. Moreover, up to the possible exclusion of $\rho_0=|0\>\<0|$, we have $\E(\D_+(\H))=\bigcup_{S\subset \RE^2} \E(\F(S)) \cup \{\rho_0\}$ (\cref{Faces_Cor}).

\medskip

\item Let $|0\>\in \H$ and define $\rho_0=|0\>\<0|$. We determine a minimal set $X_+(\H)$ of relative boundary points of $\D_+(\H)$ such that the relative interior of $\D_+(\H)$ with respect to $\A(\H)$ is given by ${\rm int}_{\A(\H)} \,\D_+(\H)= \cup_{\rho \in X_+(\H)} (\rho_0, \rho)$, where $(\rho_0,\rho)$ denotes the open line segment connecting $\rho_0$ and $\rho$. We also construct the states $\rho \in X_+(\H)$ explicitly. This result is optimal, in the sense that it fails for any proper subset $X \subsetneq X_+(\H)$. A straightforward consequence of $\D_+(\H)$ being closed is that $\D_+(\H)=\ol{\cup_{\rho \in X_+(\H)} (\rho_0, \rho)}$ for the trace norm (\cref{C18}).

\medskip

\item Moreover, the set of WPS $X_+(\H)$ can be used to generate the relative interior of $\D(\H)$. More precisely, every density matrix $\rho \in {\rm int}_{\A(\H)} \,\D(\H)$ can be written as an affine combination of $\rho_0$ and a unique state $\rho' \in X_+(\H)$ (\cref{C19}).

\end{itemize}

\medskip

\textbf{2. Infinite-dimensional case ($\H=L^2(\RE)$):}

\begin{itemize}
    
    \item The set $\D_+$ is closed but non-compact (\cref{Compact_infinite_dim}), and has an empty relative interior with respect to the affine space $\A$ (\cref{intW_+}) and to the set of density matrices $\D$ (\cref{intW_+2}), for the relative topology induced by the trace norm in $\A$ and $\D$, respectively. 
 
 \medskip
    
    \item The set $\D_{0+}$ is a dense subset of $\D_+$, i.e $\D_+=\overline{\D_{0+}}$ (\cref{X}).
\medskip
    \item The Krein--Milman relation holds for $\D_+$, i.e.~$\D_+=\overline{\mbox{\rm conv}}(\E(\D_+))$, where $\E(\D_+)$ denotes the set of extreme points of $\D_+$, and the closure is in the trace norm topology (\cref{KMR}).
\medskip    
    \item Every nonempty set $\D_+(\H_0)$, $\H_0 \subseteq L^2(\RE)$ is a face of $\D_+$. Moreover $\E(\D_+(\H_0))=\E(\D_+) \cap \D_+(\H_0)$ (\cref{Faces+Extreme}).
    
  \medskip      
    \item If $\mathcal K$ is a dense subset of $\E(\D_+)$ (with respect to the trace norm topology), then $\D_+=\overline{\mbox{\rm conv}}(\mathcal K)$ (\cref{DenseX}).

\end{itemize}
\medskip

\begin{enumerate}

\item[\textbf 3.]\textbf{Preliminary side results.} We also prove a few interesting results concerning the properties of the nodal sets of Wigner functions, which are required for the proofs of some of our main theorems. These results may also be of independent interest. They are discussed in  section 2.1 of the Preliminaries. 

\end{enumerate}

\begin{itemize}

\item Let $\rho \in \D$. If $\N(\rho)$ is bounded then the negative set $\Omega_-(\rho)$ is compact (\cref{Theorem5}). This result is used several times throughout the paper (e.g.~proofs of \cref{SecondTheorem} and \cref{rho+}).
  
  \medskip   
	
\item Fix the Hilbert space $\H\subseteq L^2(\RE)$. The nodal set $\N(\rho)$ is bounded for all $\rho\in \D(\H)$ if and only if $\H$ is a global squeezing and displacement of some Hilbert space $\H_0\subseteq \H^N$ (\cref{SecondTheorem}). 
This extends to mixed states a similar result recently obtained for pure states \cite{Abreu}. It is used in the proof of some central results of the paper (e.g.~\cref{rho+} and \cref{KMfinite}). 
  
    \medskip   
     
\item Let $\H\subseteq \H^N$ or $\H= L^2(\RE)$. If $\D_+(\H) \not=\emptyset$ then 
$\bigcap_{\rho \in \D(\H)} \N(\rho)=\emptyset$ (\cref{zeros_inters}). This result is necessary for proving \cref{ipos}.

\end{itemize}

\noindent We refer to the following sections for additional results and detailed proofs.

\subsection{Notation}\label{Notation}

The notation $\ol{A}$ denotes, depending of the context, the complex conjugate of $A\in \CO$ or the closure of the set $A$. The adjoint of the operator $A$ is denoted $A^\dagger$, and $|A|$ denotes, also depending on the context, the modulus of $A\in \CO$, the Euclidean norm of $A\in \RE^2$ or the absolute value of a trace-class operator $A\in \B_1$ (see below). We write $A \perp B$ to indicate that the objects $A$ and $B$ are orthogonal.        

The Hilbert space $L^2(\RE^n)$ is endowed with the standard inner product (which we write using the physics convention): 
$$
\<\phi|\psi\>= \int_{\RE^n} \overline{\phi(x)}\psi(x) \, dx,
$$
and the corresponding induced norm $\|\psi \|=\<\psi|\psi\>^{1/2}$. If we need to distinguish between different norms, we may also use the notation $\|\,\,\|_{L^2}$ for the $L^2$-norm.

Let $\H\subseteq L^2(\RE)$ and consider the Banach space $\B_1(\H)$ of trace-class operators over $\H$. The {\it trace} is the functional:
$$
{\rm tr}: \B_1(\H) \to \CO; \quad {\rm tr} (A)= \sum_{n} \<\phi_n|A\phi_n\> 
$$
where $(\phi_n)_n$ is an orthonormal basis of $\H$. The sum converges absolutely and is basis-independent.
The trace norm in $\B_1(\H)$ is:
$$
\| A \|_1={\rm tr}|A|
$$ 
where $|A|=\sqrt{A^\dagger A}$ is the "absolute value" of $A$, i.e.~the unique positive semi-definite operator $|A|$ such that $|A|^2=A^\dagger A $.

For a normalized $\psi \in \H\subseteq L^2(\RE)$, we denote by $\Pi_\psi$ or $|\psi\>\!\<\psi|$ the projector onto span$\{\psi\}\subseteq \H$:
$$ 
\Pi_\psi:\H\to \H; \quad \Pi_\psi \xi =\<\psi|\xi\> \psi.
$$ 
More generally, for $\phi,\psi\in \H$ the operator $|\phi\>\!\<\psi|$ is:
$$
|\phi\>\!\<\psi|:\H\to \H; \quad  |\phi\>\!\<\psi|\xi = 
\<\psi|\xi\> \phi.
$$

\section{Preliminaries}\label{Preliminaries}

\subsection{Wigner functions and their zeros}

In this section we review the definition and some of the main properties of the Wigner function, and prove a few results concerning the properties of their nodal set, which will be used in the sequel.
 
Let us introduce the general setting. We identify $\RE^2$ with the symplectic phase space $\RE_x\times \RE_\xi$ equipped with the standard symplectic form $\sigma(z,z')=Jz \cdot z'=x'\xi-x\xi'$, where $J=\left[ \begin{array}{cc} 0 & 1 \\ - 1 & 0 \end{array}\right]$ and $z=(x,\xi), \, z'=(x',\xi') \in \RE^2$. 

A matrix $S \in$ Gl$(2;\RE)$ is symplectic if $\sigma (Sz,Sz')=\sigma (z,z')$ for all $z,z' \in \RE^2$. We denote by Sp$(2)$ the group of all $2\times 2$ symplectic matrices. 
 To every $S\in $ Sp$(2)$ one can associated two {\it metaplectic} operators  
$\mu(S),\, -\mu(S)\in$ Mp$(2)$, where Mp$(2)$ denotes the metaplectic group of unitary operators on $L^2(\RE)$ \cite{GossonBook1,Shale,Weil}.

For $\psi,\phi \in L^2(\RE)$, the cross-Wigner function $W(\psi,\phi)(z)$ is a $L^2(\RE^2)$-function with support on the phase space: \cite{GossonBook1,Lions,Wigner}:
\begin{equation}\label{CWF}
W(\psi,\phi)(x,\xi)=\frac{1}{2\pi}\int_{\RE} e^{-i \xi y}\, \psi(x+\tfrac{y}{2})\overline{\phi(x-\tfrac{y}{2})}\,dy,
\end{equation}
where we have set $\hbar=1$. The cross-Wigner function is (up to a normalization factor) the Weyl symbol of the operator $|\psi\>\<\phi|$. If $\psi=\phi$ we obtain the standard Wigner function, which is a real and normalized quasi-distribution:
 \begin{equation}\label{WF}
W\psi(x,\xi)\equiv W(\psi,\psi)(z)=\frac{1}{2\pi}\int_{\RE} e^{-i \xi y}\, \psi(x+\tfrac{y}{2})\overline{\psi(x-\tfrac{y}{2})}\,dy.
\end{equation}
This definition extends naturally to mixed quantum states. Every density matrix $\rho\in \D$ is a Hilbert--Schmidt operator and therefore admits  
a kernel representation:
\begin{equation}
\rho\, \psi (x) = \int_{\RE} K_\rho(x,y)\,\psi(y)\,dy,
\end{equation}
where $K_\rho\in L^2(\RE^2)$ is the Hilbert--Schmidt kernel of $\rho$. The Wigner function $W\rho$ is  (up to a normalisation constant) the Weyl symbol of $\rho$:
\begin{equation}
W\rho(x,\xi)=\frac{1}{2\pi}\int_{\RE} e^{-i \xi y}\, K_\rho\Bigl(x+\frac{y}{2},x-\frac{y}{2}\Bigr)\,dy.
\end{equation}
A key property of Wigner functions is that the Hilbert-Schmidt inner product of two density matrices is, up to a normalization constant, the $L^2$ inner product of their Wigner functions. For every $\rho,\eta \in \D$ we have:
\begin{equation}\label{HS-inner}
{\rm tr}\, (\rho \, \eta)= 2\pi \<W\rho \, | \, W\eta \> . 
\end{equation}
Density matrices admit a spectral decomposition of the form:
\begin{equation}\label{DM1}
\rho = \sum_i \lambda_i \,\Pi_{\psi_i}, \quad {\rm where } \quad 0\le \lambda_i \le 1 , \quad \sum_i\lambda_i=1
\end{equation}
and $\Pi_{\psi_i}=|\psi_i \rangle \langle \psi_i|$ is the projector onto the ray generated by $\psi_i \in L^2(\RE)$. The first series converges in the trace norm.
Likewise, Wigner functions of mixed states can be written as 
$$
W\rho = \sum_i \lambda_i \, W\psi_i,
$$
with the series converging absolutely.

One of the most important properties of Wigner functions is their covariance under phase-space displacements and metaplectic transformations \cite{DiasRMI,GossonBook1}. For $z_0=(x_0,\xi_0) \in \RE^2$, let us define the Heisenberg-Weyl  displacement operator:
\begin{equation}\label{D(z)}
D(z_0)\psi(x)\coloneqq e^{i\xi_0(x-x_0/2)}\psi(x-x_0),
\end{equation}
which satisfies $D^\dagger(z_0)=D^{-1}(z_0)=D(-z_0)$. Then, for every $z_0\in \RE^2$ and $S \in$ Sp$(2)$, we have: 
$$
W\big(D(z_0)\mu(S^{-1})\psi\big)(z)=W\psi\big(S(z-z_0)\big)
$$
and also for mixed states:
$$
W\big(D(z_0)\mu(S^{-1})\, \rho \,  \mu(S) D^{-1}(z_0)\big)(z)= W\rho\big(S(z-z_0)\big),
$$
where $\mu(S)\in $ Mp$(2)$ is one of the two metaplectic operators associated to $S\in $ Sp$(2)$ \cite{GossonBook1,Shale,Weil}.

Recently, some of us have proved the following result, which will play a central role in this paper \cite{Abreu}:

\begin{theorem}\label{FirstTheorem}
Let $W\psi$ be the Wigner function of $\psi \in L^2(\RE)$. The nodal set $\N(W\psi)$ is bounded if and only if $\psi$ is of the form\footnote{An empty set is considered bounded.}:
$$
\psi= \sum_{n=0}^N b_n D(z_0)\mu(S^{-1}) h_n
$$
for some $N\in \IN_0$, $b_n\in \CO$, $z_0\in \RE^2$ and $S\in $ Sp$(2)$.	 Here $h_n$ is the Hermite function given by Eq.~(\ref{HermiteF}).

In other words: 
$$
\psi =D(z_0)\mu(S^{-1}) \phi \quad \mbox{ for some } \phi\in \H^N .
$$
Equivalently, $\N(W\psi)$ is bounded if and only if $W\psi$ is of the form:
$$
W\psi(z)= W\phi (S(z-z_0))
$$
for some $\phi \in \H^N$.
\end{theorem}

For $z_0 \in \RE^2$ and $S \in $ Sp$(2)$, define the Hilbert spaces 
\begin{equation}\label{Hilbert_zS}
\H_{z_0,S}^N
= D(z_0)\mu(S^{-1})[\H^N]\coloneqq\{D(z_0)\mu(S^{-1}) \phi \, |\; \phi\in \H^N \}.
\end{equation}
We can then restate Theorem \ref{FirstTheorem} in the form: 
\begin{remark}\label{FirstRemark}
Let $\psi \in L^2(\RE)$. The nodal set $\N(W\psi)$ is bounded if and only if $\psi \in \H_{z_0,S}^N$ for some $N \in \IN_0$, $z_0\in \RE^2$ and $S\in $ Sp$(2)$.
\end{remark}

\medskip

The next result relates the nodal set and the {negative set } $
 \Omega_-(\rho)\coloneqq\left\{z\in \RE^2: W\rho(z) \le 0 \right\}
 $ of $\rho \in \D$.
We have:
 
\begin{theorem}\label{Theorem5}
Let $\rho \in \D$. If the nodal set $\N(\rho)$ is bounded, then $\Omega_-(\rho)$ is also bounded and closed, hence compact.
\end{theorem}

\begin{proof}
We first prove that if $\N(\rho)$ is bounded then so is $\Omega_-(\rho)$. Since $W\rho$ is continuous and its nodal set $\N(\rho)$ is bounded, there exists $R >0$ such that for all $|z | >R$, $W\rho(z)$ is either strictly positive or strictly negative .

It follows from \cite[Theorem 7]{GossonLuef07} that  no Wigner function (pure or mixed) can be asymptotically bounded above by certain Gaussians. Hence, $W\rho(z)$ cannot be negative for all $|z | >R$, and so must be positive. We conclude that $\Omega_-(\rho)$ is bounded.

We now prove that $\Omega_-(\rho)$ is closed. Define a compact set $\Sigma \supseteq \Omega_-(\rho)$. Since $W\rho(z)$ is continuous, it attains a minimum at some point $z_0 \in \Sigma \cap \Omega_-(\rho)$. Thus $\Omega_- (\rho)$ is the preimage of the closed interval $[W\rho(z_0),0]$, and is therefore closed. Since it is also bounded, it is compact.
\end{proof}

\medskip

Let us extend Theorem \ref{FirstTheorem} to mixed quantum states.

\begin{theorem}\label{SecondTheorem}
Let $\H \subseteq L^2(\RE)$. The nodal sets $\N(\rho)$ are bounded for all $\rho \in \D(\H)$ if and only if $\H \subseteq \H_{z_0,S}^N$ for some $N \in \IN_0$, $z_0 \in \RE^2$ and $S \in $ \emph{Sp}$(2)$.
	\end{theorem}

\begin{proof}

Let us first prove sufficiency. If $\H \subseteq \H_{z_0,S}^N$ then $\N(W\psi)$ is bounded for all $\psi \in \H$ (Remark \ref{FirstRemark}). 
Let then $\rho \in \D(\H)$ be a mixed state. We have
$$
\rho =\sum_{i=0}^N \lambda_i \rho_i 
$$
for some $0\le\lambda_i\le 1$, $\sum_i \lambda_i=1$ 
and some pure states $\rho_i=|\psi_i \>\<\psi_i|$, $\psi_i\in \H$. 
Since each set $\N(W\psi_i)$ is bounded, there exists a compact $\Omega \subset \RE^2$ with $\N(W\psi_i) \subset \Omega$ for all $i=0,..,N$. By Theorem \ref{Theorem5}, $W\psi_i(z \not \in \Omega) >0$ for all $i=0,..,N$. Hence
$$
W\rho(z)= \sum_{i=0}^N \lambda_i W\rho_i (z) >0 \quad , \quad \forall z \not \in \Omega
$$
and so $\N(\rho) \subset \Omega$ is also bounded.

\medskip

Conversely, assume that $\N(W\psi)$ is bounded for every $\psi \in \H$. By Theorem \ref{FirstTheorem}, each $\psi$ belongs to some space $\H_{z_0,S}^N$ with  parameters $(N,z_0,S)$. Suppose that exist $\psi_1,\psi_2 \in \H$ such that $\psi_1 \in \H_{z_1,S_1}^{N_1}$ and $\psi_2 \in \H_{z_2,S_2}^{N_2}$ with $(z_1,S_1) \not=(z_2,S_2)$. Then $\psi =\psi_1+\psi_2 \in \H$, but $\psi \not\in \H_{z_0,S}^N$ whatever the choice of $(z_0,S,N)$. By Theorem \ref{FirstTheorem} this would imply that $\N(W\psi)$ is unbounded, which is a contradiction. Hence, we must have $(z_1,S_1) =(z_2,S_2)$ and so $\H \subseteq \oplus_{N=0}^\infty \H^N_{z_0,S}$ for some fixed $(z_0,S)$.  

Let $V=\oplus_{N=0}^\infty \H^N_{z_0,S}$. This is only a pre-Hilbert space, because it is not complete. It has an infinite countable Hamel basis $B\subseteq \cup_{N=0}^\infty B^N$ (where $B^N$ is a basis of the $(N+1)$-dimensional Hilbert space $\H^N$). Since $\H \subseteq V$, we conclude that $\H$  also has a finite, or an infinite countable Hamel basis. This, however, is only possible if $\H$ is finite-dimensional because infinite-dimensional Hilbert spaces (being complete) have uncountable algebraic dimension, i.e.~only have infinite {\it uncountable} Hamel basis \cite[Lemma 3.6.4]{Simon2015}]. Hence, we necessarily have $\H\subseteq \H^N_{z_0,S}$ for some $N\in \IN_0$, concluding the proof.
\end{proof}

\begin{remark} Observe that $\cup_{z_0,S}\H^N_{z_0,S}$ is the set of pure states of stellar rank less than or equal to $N$~\cite{chabaud2020stellar}. 
 
Moreover, the nodal sets of the states $\rho\in\D(\H^N_{z_0,S})$ are always bounded but not necessarily uniformly bounded.

\end{remark}

The next result will be used in the proof of \cref{ipos}.

\begin{theorem}\label{zeros_inters}
Let $\H\subseteq \H^N$ or $\H= L^2(\RE)$ be such that $\D_+(\H) \not=\emptyset$. Then: 
\begin{equation}\label{mathp}
\bigcap_{\rho \in \D(\H)} \N(\rho)=\emptyset.
\end{equation}
\end{theorem}

\begin{proof}
We may assume that $\H\subseteq \H^N$. Note that if the result is true for some $\H\subseteq \H^N$ such that $\D_+(\H) \not=\emptyset$, it will necessarily be true in any other case $\D(\H') \supset \D(\H)$. 

Assume that Eq.~(\ref{mathp}) is not true. Then there exists $z_0\in \RE^2$ such that: 
\begin{equation}\label{W=0}
W\psi (z_0)=0 \, , \quad \forall \psi \in \H. 
\end{equation}
Consider two arbitrary pure states $\psi_1,\psi_2 \in \H$ and construct $\psi=c_1\psi_1+c_2\psi_2$ for some $c_1,c_2 \in \CO$. From Eq.~(\ref{W=0}) we have:
$$
0=W\psi(z_0)=|c_1|^2W\psi_1(z_0)+|c_2|^2W\psi_2(z_0)+
2{\rm Re} \left(\overline{c_1}c_2W(\psi_1,\psi_2)(z_0)\right)
$$
$$
\Longrightarrow \, {\rm Re} \left(\overline{c_1}c_2W(\psi_1,\psi_2)(z_0)\right)=0 .
$$
Since this is true for all $c_1,c_2\in \CO$, we can easily conclude that:
$$
W(\psi_1,\psi_2)(z_0)=0.
$$
This equation holds for all $\psi_1,\psi_2\in \H$ and is equivalent to
\begin{equation}\label{eq_ort}
\<\psi_1| \Pi(z_0) \psi_2\>=0 \, \Longleftrightarrow \, 
\<\psi_1| D(z_0)\Pi D^\dagger(z_0) \psi_2\>=0
\end{equation}
where $\Pi(z_0)= D(z_0)\Pi D^\dagger (z_0)$ is the displaced parity operator, and $\Pi=\Pi(0)$ is the parity operator: $\Pi\psi(x)=\psi(-x)$. 
Let: 
$$
\H_{z_0}= D^\dagger (z_0)[\H]= \{D^\dagger (z_0) \psi\,|\, \psi\in \H \}.
$$ 
Equation (\ref{eq_ort}) then implies that:
$$
\<\phi_1| \Pi \phi_2\>=0 \, ,\,\, \forall \phi_1,\phi_2 \in \H_{z_0} \Longleftrightarrow \H_{z_0} \perp \Pi\H_{z_0} \Longleftrightarrow \D(\H_{z_0}) \perp \D(\Pi\H_{z_0}),
$$
where $\Pi\H_{z_0}= \{\Pi \psi\,|\, \psi\in \H_{z_0} \}$. It follows that (cf.~Eq.~(\ref{HS-inner}): 
$$
\int  W\rho_1(z) W\rho_2(z) dz=0 \, , \quad \forall \rho_1\in \D(\H_{z_0})\, ,\, \rho_2\in \D(\Pi\H_{z_0}).
$$
Now assume that $\rho \in \D_+(\H)$. Then $\rho_1 = D^\dagger(z_0) \rho D(z_0) \in \D_+(\H_{z_0})$, and  $\rho_2=\Pi \rho_1 \Pi \in D_+(\Pi\H_{z_0})$. We then have (using the continuity of the Wigner function):   
$$
\int W\rho_1(z) W\rho_2(z) dz=0 \quad \Longrightarrow 
\quad \N(\rho_1) \, \cup \, \N(\rho_2) =\RE^2,
$$
which is not possible since both nodal sets are bounded ($\N(\rho_2)=-\N(\rho_1)\coloneqq\{-x \,|\, x\in \N(\rho_1)\}$ and $\H_{z_0} \subseteq \H_{z_0}^N \Longrightarrow \rho_1 \in \D(\H_{z_0}^N) \Longrightarrow  \N(\rho_1)$ is bounded). We conclude that Eq.~(\ref{mathp}) must hold. 
\end{proof}

\subsection{Convex analysis}

Our notation and results follow closely the presentation of \cite[Chapter 8]{Simon2011}, \cite[Chapter 2]{Webster} and \cite{Border2,Border5}. In this subsection  and the next, $V$ is a real topological vector space (TVS). Recall the basic definitions:

\begin{definition}\label{Aff+Conv}
Let $S\subseteq V$ be a subset of a real TVS $V$. Then $S$ is {\bf(A)} affine, {\bf(B)} convex if and only if $S$ is nonempty, and for all $u,v \in S$:
\begin{enumerate}
	\item[{\bf(A)}] $tu+(1-t)v \in S$, $\forall t \in \RE$
	\smallskip
	\item[{\bf(B)}] $tu+(1-t)v \in S$, $\forall t \in [0,1]$
\end{enumerate}
\smallskip
The affine hull, convex hull and closed convex hull of a subset $X \subseteq V$, denoted ${\rm Aff} (X)$, ${\rm conv} (X)$ and $\overline{\rm conv} (X)$ respectively, are defined as the smallest affine, convex and closed convex subsets of $V$ that contain $X$. 

Equivalently, ${\rm Aff} (X)$ and ${\rm conv} (X)$ can be defined as the sets of all affine, respectively convex, combinations of the elements of $X$ (cf.~\cite[Theorem 5.2]{Simon2011}):
\begin{equation}\label{AffineHull}
{\rm Aff}(X)= \left\{ \sum_i t_i x_i : \, x_i \in X,\, t_i \in \RE, \, \sum_i t_i =1 \right\}
\end{equation}
\begin{equation}\label{ConvexHull}
{\rm conv}(X)= \left\{ \sum_i t_i x_i : \, x_i \in X,\, t_i \in [0,1], \, \sum_i t_i =1 \right\}
\end{equation}
where the sums are finite. We also always have $\overline{\rm conv} (X)=\overline{{\rm conv} (X)}$ \cite[Corollary 5.1.2]{Border5}.

Finally, for $u,v \in V$, the line segment from $u$ to $v$ is written: 
$$
[u,v]={\rm conv}(\{u,v\})=\{(1-t)u+tv: t\in [0,1] \}
$$
and $(u,v)=[u,v]\,\backslash \, \{u,v\}$, $(u,v]=[u,v]\,\backslash \, \{u\}$, $[u,v)=[u,v]\,\backslash \, \{v\}$.
\end{definition}

Note that if $V$ is finite-dimensional and $X$ is compact then ${\rm conv} (X)$ is also compact; and so $\overline{\rm conv} (X)={\rm conv} (X)$ \cite[Corollary 2.2.4]{Border2}.  

Two other key definitions in convex analysis are those of {\it face} and {\it extreme point} \cite[p.162]{Rockafellar}. 

\begin{definition}\label{extreme}
Let $X\subseteq V$ be a convex subset of a vector space $V$. 
\begin{itemize}
\item A convex set $\F \subseteq X$ is a face of $X$ if and only its points cannot be written as a proper convex combination of elements of $X \backslash \F$, i.e
$$
w\in \F \quad \mbox{and} \quad w=(1-t)u+tv , ~ 0<t<1 , ~  u,v \in X \,\Longrightarrow \, u,v \in \F
$$
If $\F \not= X$ then $\F$ is a proper face.

\item Extreme points are faces with a single point, i.e.~$w \in \E(X) $ if and only if
$$
w=(1-t)u+tv , ~ 0<t<1 , ~  u,v \in X \,\Longrightarrow \, u=v=w
$$	
\end{itemize}

\end{definition}

Note that if $\F$ is a face of a convex set $X$ then \cite[Proposition 8.6]{Simon2011}:
\begin{equation}\label{ExtFaces}
\E(\F)=\E(X)\cap \F
\end{equation}

Faces can be determined as extreme sets of continuous functionals. Let $f: V \longrightarrow \RE$ be a linear functional that is not constant on the convex  set $X \subset V$. If $a=$ sup $\{f(x): x\in X\}$ then
\begin{equation}\label{Faces}
\F= \{x\in X : f(x)=a \}
\end{equation}
if nonempty, is a proper face of $X$ \cite[Theorem 8.3]{Simon2011}.

One of the main theorems in convex analysis is the Krein--Milman Theorem \cite{Krein,Simon2011} that we have already stated in the Introduction (\cref{Krein--Milman}). In the finite-dimensional case, if $S \subset V$ is convex and compact, the Krein--Milman Theorem simplifies to \cite[Theorem 2.6.16]{Webster}:
\begin{equation}\label{KM-eq-finite}
S={\rm conv}(\E(S) )		
\end{equation}
We shall discuss the infinite-dimensional case in more detail in \cref{sec:geom_infinite}.

\subsection{Relative topology}

In relative topology the notions of interior and boundary are dependent of the ambient space. In this section we provide some standard definitions and a few basic results. The reader should refer to \cite{Border2,Border5,Rockafellar,Simon2011,Webster} for more details.

\begin{definition}
Let $X$ be a subset of a real TVS $V$. Let $S \subset X$.
Then:

\begin{enumerate}

\item[{\bf(A)}] The relative interior of $S$ with respect to $X$, denoted ${\rm int}_X S$, is defined by: 
$$
x\in \, {\rm int}_X S \, \Longleftrightarrow \,\mbox{there exists an open set }\, U\subset V \mbox{ with }x\in U  \cap X \subseteq S
$$

\item[{\bf(B)}]  The relative boundary of $S$ with respect to $X$, denoted $\partial_X S$, is defined as:
$$
\partial_X S =X\setminus ({\rm int}_X S \, \cup \, {\rm int}_X (X\backslash S))
$$

\item[{\bf(C)}] If $X={\rm Aff}(S)$ then ${\rm int}_X S={\rm iint} S$ and $\partial_X S =\partial^i S$ (the intrinsic interior and intrinsic boundary, respectively). 
\smallskip
\item[{\bf(D)}] Let $X$ be an affine subspace of $V$ of dimension $\ge 1$. The relative algebraic interior of $S$ with respect to $X$, denoted ${\rm aint}_X\, S$, is defined by (cf.~\cite{Weis1} and \cite[pag.2]{Zalinescu}):
$$
x\in \, {\rm aint}_X S \, \Longleftrightarrow \, \forall_{y\in X\backslash \{x\}} \, \exists_{z\in (x,y]}: [x,z] \subset S~.
$$

\end{enumerate}

\end{definition}

\begin{notation} The above definitions of relative interior and relative boundary are the definitions of interior and boundary in the relative (induced) topology in $X$. In convex analysis these terms are often used to denote the intrinsic interior and boundary (see, for example, \cite{Border5,Rockafellar,Webster}). However, we need the more general definitions given in {(A)} and {(B)} above, hence the more explicit notation. In (C) we use the notation of \cite{Simon2011}. The definition given in (D) is more general than the one in \cite{Weis1} (where $X={\rm Aff}(S)$). An equivalent definition of relative algebraic interior is given in \cite[Definition 5.2.5]{Border5}.  
\end{notation}

 
 In the sequel we will also use the following Theorem:

\begin{theorem}\label{ai}
Let $X$ be an affine subspace of a finite-dimensional TVS $V$ and let ${\rm dim}\, X \ge 1$. Let $S\subset X$ be a convex set. We have:
\begin{equation}\label{AI}
{\rm int}_X S = {\rm aint}_{X} S
\end{equation}
\end{theorem}

\begin{proof}
This Theorem is well-known for the case $X=V$ \cite[Corollary 6.4.1]{Rockafellar} and also for $X={\rm Aff}(S)$ \cite[Corollary 3.3]{Weis2}, (see also \cite[Proposition 5.2.7]{Border5}).

We have not found our slightly more general version in the literature, so we give a simple proof:
Since $S\subset X$, and $X$ is an affine space, we necessarily have $ {\rm Aff}(S)\subseteq X $ (recall that ${\rm Aff}(S)$ is the smallest affine space that contains $S$). Hence, the only case that we still have to prove is ${\rm Aff}(S)\subsetneq X$. It follows from \cref{Simon}-A that in this case ${\rm int}_X S=\emptyset$ (see below). On the other hand, since ${\rm Aff}(S)\subsetneq X$ there exists $y \in X\backslash {\rm Aff}(S)$, and so:
$$   
(x,y) \cap S =\emptyset \, , \quad \forall x \in S,
$$ 
(otherwise, $y \in {\rm Aff}(S)$). This, in turn, implies that $x \notin {\rm aint}_{X} S$, $\forall x\in S$ and thus ${\rm aint}_{X} S=\emptyset$. Therefore Eq.~(\ref{AI}) also holds trivially in the case ${\rm Aff}(S)\subsetneq X$, concluding the proof.
\end{proof}

We will also need the following basic results from convex analysis:
 
 \begin{lemma}\label{Simon} 
 Let $X$ be an affine subspace of a finite-dimensional TVS $V$. Let $S\subset X$ be a nonempty convex set. In {\bf (B, C, D)} assume that dim $X\ge 1$.  We have:
\begin{enumerate}

\item[{\bf (A)}] ${\rm int}_X S \not= \emptyset $ if and only if $X={\rm Aff}(S)$.
\smallskip

\item[{\bf (B)}] If $x \in {\rm iint}\, S$ and $y \in S\backslash \{x\}$ then $[x,y) \subset {\rm iint} \,S$.
\smallskip
\item[{\bf (C)}] Let $x \in {\rm iint}\, S$. Then $\forall_{y \in S \backslash\{x\}} ~\exists_{z\in S}:~ x\in (y,z)$.
\smallskip
\item[{\bf (D)}] If $S$ is compact and $x \in {\rm iint}\, S$, we have: 
$$  \forall_{y \in S \backslash\{x\}} \,\, \exists^1_{z\in ~{\partial^i} S}:~ x\in (y,z).
$$
\smallskip
\end{enumerate}
 	
\end{lemma}

\begin{proof} 

{\bf (A)} From \cite[Theorem 8.8]{Simon2011} there is a unique affine space $X \supset S$ such that ${\rm int}_X S \not= \emptyset $. In \cite[Theorem 6.2]{Rockafellar} is proved that this space is $X={\rm Aff}(S)$.

\medskip

{\bf (B)} Proved in \cite[Lemma 2.3.3]{Webster}.

\medskip

{\bf (C)} A simple proof follows from Theorem \ref{ai} with $X={\rm Aff}(S)$.  Let $y \in S\backslash \{x\}$ and define $u=2x-y$; then $u\in X$. Moreover, $x \in   {\rm aint}_{X} S$ and thus  exists $z\in (x,u]$ such that $z\in S$. Let $z=(1-p) x + p u$ for some $0 < p \le 1$, and consider the interval:
$$
[y,z]= \{(1-t)y +t z , \quad t \in [0,1] \}.
$$
Setting $t=\tfrac{1}{1+p}$ we show that $x\in (y,z)$.
 
\medskip

{\bf (D)} Consider the affine subspace: ${\rm Aff}(\{x,y\})$. This is a straight line, hence convex and closed. Since $S$ is compact, we have ${\rm Aff}(\{x,y\}) \,\cap \, S =[u,v]$ for some $u,v \in S$. 

From (C) we also conclude that $u,v \in \partial^i S$  (indeed, if say $v \in {\rm iint} \,S$, then there exists $z\in S$ such that $[u,v] \subsetneq [u,z) \subsetneq  {\rm Aff}(\{x,y\})\cap S=[u,v]$, a contradiction).   

Since $x\in {\rm iint}\, S$, we must have $x\in (u,v)$. Note also that $y\in [u,v]$. The two cases $y\in [u,x)$ and $y\in (x,v]$ are symmetric, so assume $y\in [u,x)$. If we take $z=v\in \partial^iS$ then $x\in (y,z)$, which proves existence.

Conversely, if $z\in S$ and $x\in (y,z)$ then we necessarily have $z\in {\rm Aff}(\{x,y\}) \,\cap \, S =[u,v]$. Since $y \in [u,x)$ we conclude that $z\in (x,v]$. By (B) the only boundary point of $S$ in $(x,v]$ is $v$, so $z=v$, proving uniqueness.   
\end{proof}

\section{Topology}
\label{sec:topo}

We consider the chain of sets: $\D_+(\H) \subset \D(\H) \subset \A(\H) \subset \B_1(\H) $ and investigate the topological properties of the sets of WPS. For the rest of this section the topology is always the relative topology induced by the trace norm in the affine hull of $\D(\H)$.

\subsection{The set of quantum states}

In this section, we give results on the properties of $\D(\H)$ which we use in the following sections. Most of these results are well-known, but we provide proofs for completeness. We consider both the finite-dimensional and infinite-dimensional cases, i.e.~$\H \subseteq \H^N$ and $\H=L^2(\RE)$.
Unless otherwise stated, the results are valid for both cases. 

Our first result concerns the affine hull of $\D(\H)$.

\begin{theorem}
Let $\H \subseteq L^2(\RE)$ and consider the affine space  
$$
 \A(\H)=\left\{\rho \in \B_1(\H):\, \rho \mbox{ is self-adjoint } , \; {\rm tr}\,\rho=1 \right\}
$$
We have: $\A(\H)={\rm Aff}(\D(\H))$. Moreover, $\A(\H)$ is closed.
\end{theorem}

\begin{proof}
It is obvious that $\D(\H) \subset \A(\H)$, and since $\A(\H)$ is an affine space we also have {\rm Aff}$(\D(\H))\subseteq \A(\H)$ (since {\rm Aff}$(\D(\H))$ is a subset of all affine spaces containing $\D(\H)$).

Conversely, let $\rho \in \A(\H)$. Since $\rho$ is self-adjoint and trace-class (hence compact) we can write it in the form $\rho=\eta_-+\eta_+$ where $\eta_-,\eta_+$ are both self-adjoint, trace-class and $\eta_-\le 0$, $\eta_+ \ge 0$. Let $a_{\pm}={\rm tr}$ $ \eta_{\pm}$ and define $\rho_+=\eta_+/ a_+$ and $\rho_-=\eta_-/ a_-$ (assuming that $a_{-} \not=0$). Then
$$
\rho = \eta_-+ \eta_+ = a_-\rho_- + a_+ \rho_+  
$$
where $\rho_{\pm} \ge 0$ and ${\rm tr}$ $\rho_{\pm}=1$. Moreover:
$$
{\rm tr}\, \rho =1 \, \Longrightarrow \, a_-+a_+=1
$$
Hence, $\rho  $ is an affine combination of $\rho_-,\rho_+ \in \D(\H)$ and so $\rho \in $ Aff$(\D(\H))$. If $a_{-} =0$ then $\rho=\eta_{+} = \rho_{+} \in \D(\H) \subset$ Aff$(\D(\H))$, and we conclude that $\A(\H)={\rm Aff}(\D(\H))$.

To prove that $\A(\H)$ is closed, note that $\A(\H) \subset \B_1(\H)$ and that $\B_1(\H)$ is a Banach space, hence closed with respect to the trace norm. Let $(\rho_n)_{n\in \IN} \subset \A(\H)$ be such that $\rho_n \xrightarrow[]{\|\, \, \|_1} \rho \in \B_1(\H)$. Since $\|\rho_n-\rho \|_1=\|\rho_n^\dagger-\rho^\dagger \|_1$, we also have $\rho_n^\dagger \xrightarrow[]{\|\, \, \|_1} \rho^\dagger $. Moreover, $\rho_n^\dagger=\rho_n$ for all $n\in \IN$, and so $\rho^\dagger =\rho$. Hence, $\rho$ is self-adjoit. Since the trace is continuous with respect to the trace norm, we also conclude that tr$\,\rho=1$. It follows that $\rho \in \A(\H)$ and so $\A(\H)$ is closed.   
\end{proof}

\begin{theorem}\label{Compact}
Let $\H\subseteq L^2(\mathbb R)$. The set $\D(\H)$ is convex, closed and bounded. It is  compact if and only if $\H$ is finite-dimensional.
\end{theorem}

\begin{proof} 
 It is trivial that $\D(\H)$ is convex.

Let us prove that $\D(\H)$ is closed. Recall from the previous theorem that $\A(\H)$ is closed, and let $(\rho_n)_{n\in \IN}$ be a sequence in $\D(\H)$ converging to $\rho \in \A(\H)$ in the trace norm.

Let $v \in \H$. Then $\<v| \rho_n v\> \ge 0$ for all $n$, and since the inner product is continuous in the $L^2$-norm, and $\rho_n \xrightarrow[]{\|\, \, \|_1} \rho$ implies $\rho_n v \xrightarrow[]{L^2} \rho v$ for all $v\in \H$, we also have $\<v| \rho v\>  \ge 0$.

Hence $\rho\in \A(\H)$ is positive semi-definite and so $\rho \in \D(\H)$. We conclude that $\D(\H)$ is closed. In particular, $\D(\H)$ is a complete metric space since $\B_1(\H)$ is a Banach space.

It is also immediate that $\D(\H)$ is bounded, so in the case where $\H$ is finite-dimensional, we conclude that $\D(\H)$ is compact using the Heine--Borel theorem.

On the other hand, in the infinite-dimensional case, consider the sequence of density operators defined by $\rho_n\coloneqq\ket n\!\bra n$, $n\in \IN$.
For all $m\neq n$ we have $\|\rho_m-\rho_n\|_1=2$, and thus $(\rho_n)_{n\in\IN}$ is a sequence in $\D(\H)$ with no Cauchy subsequence and therefore no converging subsequence. Hence, $\D(\H)$ is not compact since it is complete.
\end{proof}
Another interesting example of a sequence of density operators with no Cauchy subsequence can be defined via a recurrence relation: 
$$
    \rho_0\coloneqq\ket0\!\bra0,\quad\rho_{n+1}\coloneqq\frac12\ket{n+1}\!\bra{n+1}+\frac12\rho_n , \quad n \in \IN_0. 
$$
We now investigate the properties of the interior points of $\D(\H)$:
\begin{theorem} \label{intQS}
\begin{enumerate}
\item[{\bf(A)}] If $\H \subseteq \H^N$ then $\rho \in {\rm iint} \, \D(\H)$ if and only if $\rho$ is of maximal rank.  
\item[{\bf(B)}] If $\H=L^2(\RE)$ then ${\rm iint} \, \D(\H)=\emptyset$.
\end{enumerate}

\end{theorem}

\begin{proof}

{\bf(A)} If ${\rm dim}\, \H=1$ the result follows trivially from the fact that $\D(\H)=\A(\H)={\rm iint}\, \D(\H)= \{|v\>\<v| \}$ for some normalized $v\in \H$, and that $|v\>\<v|$ is of maximal rank in $\H$. 

Let then dim$\,\H=n >1$, and 
let $\rho \in \D(\H)$ be of maximal rank. Then $\rho$ can be written in the form
\begin{equation}\label{eigenexp}
\rho =\sum_{i=1}^n \lambda_i \rho_i
\end{equation}
where $\lambda_i >0$ for all $i=1,...,n$, and $\sum_{i=1}^n \lambda_i=1$. Moreover, $\rho_i=|v_i\>\<v_i|$ with $v_i\in \H$ an orthonormal basis of eigenvectors of $\rho$.

We want to show that there exists $C>0$ such that 
\begin{equation}\label{minimaldistance}
\| \rho - \sigma\|_1 >C \quad , \quad \forall 
\sigma \in \A(\H)  \backslash \D(\H). 
\end{equation}
We have for all $u \in \H,\, \|u\|_{L^2}=1$
$$
\|\rho-\sigma\|_1 \ge \| (\rho-\sigma) u\|_{L^2}~.
$$
Let $\lambda_-(\sigma) <0$ be a negative eigenvalue of $\sigma$, and let us set $u$ to be the corresponding (normalised) eigenvector: $\sigma u =\lambda_-(\sigma) u$. We have 
$$
\| (\rho-\sigma) u\|_{L^2} = \|(\rho -\lambda_-(\sigma)) u\|_{L^2} = \| \sum_{i=1}^n(\lambda_i-\lambda_-(\sigma)) k_i v_i\|_{L^2} 
$$
$$
=\left( \sum_{i=1}^n(\lambda_i-\lambda_-(\sigma))^2 |k_i|^2 \right) ^{1/2} > \lambda_{\rm min}(\rho)
$$
where $k_i= \<v_i|u\>$ and $\lambda_{\rm min}(\rho)>0$ is the smallest eigenvalue of $\rho$. Hence, (\ref{minimaldistance}) holds for all $0< C \le \lambda_{\rm min}(\rho)$.  
Finally, if we set $\delta =C$, we get
$$
B_\delta(\rho) \cap \A(\H) = B_\delta(\rho) \cap \D(\H) \subset \D(\H) \Longrightarrow \rho \in {\rm iint} (\D(\H)) \, ,
$$
where $B_\delta(\rho)=\{\eta \in \B_1: \|\eta -\rho \|_1 < \delta \}$. This completes the first part of the proof.

We now assume that $\rho$ is not of maximal rank, and define:
$$
\sigma_t=(1+t)\rho - t \Pi_v \quad , \quad t>0
$$ 
where $\Pi_v$ is a projector which is orthogonal to $\rho$, i.e.~tr$\, (\rho \Pi_v)=0$ (which always exists because $\rho$ is not of maximal rank). It follows that $\sigma_t \in \A(\H)  \backslash \D(\H)$, for all $t>0$, but:
$$
\|\sigma_t - \rho\|_1 =|t| \,\|\rho-\Pi_v\|_1 \to 0
$$   
as $t\to 0$. Hence, for all $\delta>0$, $B_\delta(\rho) \cap \A(\H)  \not\subset \D(\H)$, and so $\rho \notin {\rm iint}\D(\H)$.\\

{\bf(B)} Let now $\H=L^2(\RE)$, and let $\rho \in \D$. Then $\rho=\sum_{i=1}^n \lambda_i\rho_i$ for some orthogonal pure state density matrices $\rho_i$, and some nonnegative sequence $(\lambda_i)_{ i\in \IN}$ such that $\sum _{i=1}^n \lambda_i=1$. If $n <\infty$ then $\rho \not\in {\rm iint}\,\D$ (the proof follows exactly the steps of the finite-dimensional case).

If $n=\infty$, consider the sequence:
$$
\sigma_k=\frac{1}{n_k}\sum_{i=1}^k\lambda_i\rho_i - \frac{1}{n_k}\sum_{i=k+1}^\infty   
\lambda_i\rho_i
$$
where $n_k=\sum_{i=1}^k\lambda_i - \sum_{i=k+1}^\infty\lambda_i$ is strictly positive for sufficiently large $k$.
 Then $\sigma_k\in \B_1$, $\sigma_k$ is self-adjoint and ${\rm tr}\, \sigma_k=1$, but $\sigma_k \notin \D$. Therefore $\sigma_k\in \A  \backslash \D$. Moreover:
\begin{eqnarray*}
\|\rho-\sigma_k\|_1 &\le &|1-\tfrac{1}{n_k}|\left|\left|\sum_{i=1}^k\lambda_i\rho_i\right|\right|_1 + |1+\tfrac{1}{n_k}|\left|\left|\sum_{i=k+1}^\infty\lambda_i\rho_i\right|\right|_1 \\
& < & |1-\tfrac{1}{n_k}| \sum_{i=1}^k\lambda_i + |1+\tfrac{1}{n_k}| \sum_{i=k+1}^\infty\lambda_i  \to 0
\end{eqnarray*} 
as $k\to \infty$ (note that $n_k \to 1$ as $k\to \infty$). Hence $\rho \notin {\rm iint}\,\D$, concluding the proof.
\end{proof}

\subsection{Topology of Wigner-positive states in finite dimension}

In this section, we focus on the finite-dimensional case.

\begin{theorem}\label{Compact_finite_dim}
If $\H$ is finite-dimensional, the set $\D_+(\H)$, if nonempty, is convex and compact.
\end{theorem}

\begin{proof} 

It is obvious that $\D_+(\H)$ is convex. It is also bounded since $\D_+(\H) \subset \D(\H)$.

Let us now prove that $\D_+(\H)$ is a closed subset of $\D(\H)$, and thus is itself closed since $\D(\H)$ is closed (see \cref{Compact}). Let $\rho_n \in \D_+(\H)$, $n\in \IN$ be a sequence converging to $\rho \in \D(\H)$ in the trace norm. Then $\rho_n \to \rho$ in the Hilbert--Schmidt norm and thus $W\rho_n \to W \rho$ in $L^2$. Suppose $\rho \notin \D_+(\H)$. Since $W\rho$ is continuous, there exists an open set $\Omega \subset \RE^2$ such that $W\rho(z \in \Omega) <0$ and thus:
$$
\int_{\Omega} |W\rho(z)|^2 d^2z \ge  C 
$$
for some $C>0$. It follows that:
$$
\|W\rho_n - W\rho\|_2^2 \ge \int_{\Omega} |W\rho_n(z)-W\rho(z)|^2 d^2z
\ge  \int_{\Omega} |W\rho(z)|^2 d^2z \ge  C 
$$
for all $n$, and so $W\rho_n \not\to W\rho $ in $L^2$. Hence, we must have $\rho \in \D_+(\H)$, and so $\D_+(\H)$ is closed.

We conclude that $\D_+(\H)$ is bounded and closed. Since it is also finite-dimensional, it is compact.
\end{proof}


The next theorem characterises the interior points of $\D_+(\H)$ relative to $\A(\H)$ for $\H\subseteq\H^N$. First we prove the following preparatory lemma.

\begin{lemma}\label{rho+}
For $\H\subseteq \H^N$, let $\rho_0 \in \D_+(\H)$ be such that $\N(\rho_0) = \emptyset$\footnote{We do not assume that $\rho_0=|0\>\<0|$.}. Let $\rho_1 \in \D(\H)\backslash \D_+(\H)$. 
Then there is one, and only one, $\rho \in (\rho_0,\rho_1)\cap \D_{0+}(\H)$.
The state $\rho$ is given explicitly by:
\begin{equation}\label{eqrho+}
\rho = (1-t_0) \rho_0+t_0\rho_1, \quad {\rm where} \quad
t_0 =\frac{1}{1-k_0}, \quad k_0= {\rm min}\left\{ \tfrac{W\rho_1(z)}{W\rho_0(z)} \, , \, z\in \RE^2\right\}
\end{equation}
Moreover:
\begin{enumerate}
\item[{\bf  (A)}] $[\rho_0,\rho) \subset \D_+(\H) \backslash \D_{0+}(\H)$
\smallskip
\item[{\bf  (B)}] $(\rho,\rho_1] \subset \D(\H) \backslash \D_+(\H)$ 
\smallskip
\item[{\bf  (C)}] $t_0=\frac{\|\rho-\rho_0\|_1}{\|\rho_1-\rho_0\|_1}$

\end{enumerate}

\end{lemma}
\smallskip
\begin{proof}
By construction $\rho\in (\rho_0,\rho_1)$. Let us show that $\rho \in \D_{0+}(\H)$. 

Recall that the set $\Omega_-(\rho_1)=\{z\in \RE^2: W\rho_1(z) \le 0 \}$ is compact (cf.~Theorem \ref{Theorem5} and \cref{SecondTheorem}). 
Consider the convex combinations:
$$
\rho_t=(1-t)\rho_0 + t \rho_1\, , \quad t\in [0,1]
$$
and for each $z\in \Omega_-(\rho_1)$ define the profile function $t \to f_z(t)=W\rho_t(z)$. This function is strictly decreasing and continuous. It has a single zero at:
\begin{equation}\label{t(z)}
f_z(t)=0 \Longleftrightarrow t=\frac{W\rho_0(z)}{W\rho_0(z)-W\rho_1(z)}. 
\end{equation}
Let us denote the right-hand side of the previous equation by $t(z)$. Since $t(z)$ is defined on the compact set $\Omega_-(\rho_1)$ and is continuous, it attains a minimum at some $z_0 \in \Omega_-(\rho_1)$:
$$
t_0= t(z_0)= \min \{t(z)| z\in \Omega_-(\rho_1) \}\, .
$$
Thus $t_0$ is the smallest value of $t$ for which $W\rho_t$ has a zero. For $z\not=z_0$ we have $t(z)\ge t_0$ and so $W\rho_{t_0}(z) \ge 0$ for $z\in \RE^2$. Moreover, $W\rho_{t_0}(z_0)=0$ and so $\rho_{t_0} \in \D_{0+}(\H)$. It also follows from (1) and (2) below that there is no other $\rho \in [\rho_0,\rho_1]\cap \D_{0+}(\H)$. 

Let us go back to Eq.~(\ref{t(z)}). We easily conclude that $t(z)=1/(1-k(z))$ for $k(z)= \frac{W\rho_1(z)}{W\rho_0(z)}$, and that the minimum of $t(z)$ is reached at the minimum of $k(z)$, $z\in \RE^2$. Denote the latter minimum by $k_0$. Then $k_0 <0$ and $t_0=1/(1-k_0)$, proving Eq.~(\ref{eqrho+}).   

Since the functions $f_z(t)$ are strictly decreasing, we also conclude that:

\smallskip

\noindent
{\bf(A)} For $t<t_0$, we have $f_z(t) =W\rho_t(z) >0,\,\forall z$ and so $\rho_t \in \D_+(\H)$ and $\N(\rho_t)=\emptyset$.

\smallskip

\noindent
{\bf(B)} For $t>t_0$, we have $W\rho_{t}(z_0)<0$ and so $\rho_{t} \notin  \D_+(\H)$.

\smallskip

\noindent
{\bf(C)} Finally, taking norms in $\rho=(1-t_0)\rho_0+t_0\rho_1$ easily leads to (C). 
\end{proof}

We then have:

\begin{theorem}\label{ipos}
 Let $\H\subseteq \H^N$, and let $\rho \in \D(\H)$. Then $\rho \in {\rm int}_{\A(\H)}\D_+(\H)$ if and only if $\rho$ is of maximal rank and $\N(\rho) = \emptyset$.
\end{theorem}

\begin{proof}

We consider first the trivial case ${\rm dim}\,\H=1$. If $\H=\H^0$ then $\D_+(\H)=\A(\H)=\{\rho_0\}$, where $\rho_0=|0\>\<0|$, and so ${\rm int}_{\A(\H)}\D_+(\H)=\{\rho_0\}$. Since $\rho_0 : \H^0 \to \H^0$ is of maximal rank and $\N(\rho_0) =\emptyset $, the result holds. Moreover, if $\H={\rm span}\{v\}$ for some $v\in \H^N$, $v\not=|0\>$ then, from Hudson's Theorem, $\N(|v\>\<v|)\not=\emptyset$ and since $\D_+(\H)\subseteq \D(\H)=\{|v\>\<v|\}$, we have ${\rm int}_{\A(\H)}\D_+(\H)=\D_+(\H)=\emptyset$, and thus the result also holds.    

Let then ${\rm dim} \, \H \ge 2$. Assume that $\rho$ is of maximal rank. From \cref{ai} and \cref{intQS}:
$$
\rho \in {\rm iint} \, \D(\H)= {\rm aint}_{\A(\H) } \, \D(\H)
$$
and thus
$$
\forall_{\sigma \in \A(\H) } ~ \exists_{\rho_1\in (\rho,\sigma]}:\, [\rho,\rho_1] \subset \D(\H) \, .
$$
Now, suppose that $\N(\rho)= \emptyset$. From \cref{rho+} there exists $\rho_2 \in (\rho,\rho_1]$ such that $[\rho,\rho_2] \subset \D_+(\H)$ (note that \cref{rho+} applies only to the case $\rho_1 \in \D(\H) \backslash \D_+(\H)$; but if $\rho_1\in \D_+(\H)$ we have trivially $\rho_2=\rho_1$). 

Putting the two results together, we conclude that:
$$
\forall_{\sigma \in \A(\H) } \, \exists_{\rho_2\in (\rho,\sigma]}:\, [\rho,\rho_2] \subset \D_+(\H) \Longleftrightarrow \rho \in  {\rm aint}_{\A(\H) } \, \D_+(\H) ={\rm int}_{\A(\H) } \, \D_+(\H)  
$$
where the last equality again follows from Theorem \ref{ai} and the fact that $\D_+(\H)$ is convex. 

Conversely, if the rank of $\rho$ is not maximal then, from Theorem \ref{intQS}, $\rho \notin {\rm iint} \, \D(\H) \supset {\rm int}_{\A(\H) } \, \D_+(\H)$.

Finally, suppose that $\rho \in  {\rm int}_{\A(\H) } \, \D_+(\H)$. Then ${\rm int}_{\A(\H)}\D_+(\H) \not=\emptyset$ and so (cf.~Lemma \ref{Simon}-A) $\A(\H)={\rm Aff}(\D_+(\H))$ and $\rho \in {\rm iint} \, \D_+(\H)$. If we choose an arbitrary $\eta \in \D_+(\H) \backslash \{\rho\}$, from Lemma \ref{Simon}-C there exists $\eta' \in \D_+(\H)$ such that $\rho \in (\eta,\eta')$. It follows that $\N(\rho)= \N(\eta) \cap \N(\eta')$, and thus $\N(\rho) \subseteq \N(\eta)\, , \forall \eta \in \D_+(\H)$. Since $\D(\H) \subset$ Aff$(\D_+(\H))$, if $z_0 \in \N(\rho) = \bigcap_{\eta \in \D_+(\H)} \N(\eta)$ then $z_0 \in \N(\eta)$ for all $\eta \in \D(\H)$ (every affine combination of two elements with a zero at $z_0$ will have a zero at $z_0$), which contradicts \cref{zeros_inters}. We therefore conclude that $\N(\rho)=\emptyset$, completing the proof. 
\end{proof}

\medskip

The following is a trivial corollary of the previous theorem. It provides a characterisation of the relative boundary of $\D_+(\H)$.

\begin{corollary}\label{boundary}
Let $\H \subseteq \H^N$. The relative boundary of $\D_+(\H)$ with respect to $\A(\H) $ is given by:
$$
\partial_{\A(\H)} \D_+(\H) = \D_{0+}(\H) \cup \left(\D_+(\H) \cap \partial^i\D(\H) \right) 
$$
More precisely:
\begin{enumerate}

\item[{\bf(A)}] $\rho\in \partial_{\A(\H)} \D_+(\H) \cap {\rm iint}\, \D(\H) \Longrightarrow \rho \in \D_{0+}(\H) $
\smallskip
\item[{\bf(B)}] $\rho\in \partial_{\A(\H)} \D_+(\H) \cap \partial^i\D(\H) \Longleftrightarrow \rho \in \D_+(\H) \cap \partial^i\D(\H)$.

\end{enumerate} 

\end{corollary}

We now determine the Hilbert spaces for which the sets $\D_+(\H)$ are nonempty or have a nonempty interior. 

\begin{theorem}\label{intWN}

Let $\H \subseteq \H^N$. Then:
\begin{enumerate}

\item[{\bf(A)}] If $|0\rangle  \perp \H$ then $\D_+(\H) = \emptyset$.
    \smallskip
\item[{\bf(B)}] If $|0\rangle \in \H$ then ${\rm int}_{\A(\H) } \D_+(\H)$ is not empty. Consequently, ${\rm Aff}(\D_+(\H))={\rm Aff}(\D(\H))=\A(\H) $, and ${\rm int}_{\A(\H) }\, \D_+(\H) = {\rm iint}\, \D_+(\H)$.
\smallskip
\item[{\bf(C)}] If $|0\rangle  \not\in \H$ but $|0\rangle \not\perp \H$, then the two previous cases are possible, i.e.~we can have  $\D_+(\H) = \emptyset$ or ${\rm int}_{\A(\H)}\D_+(\H) \not= \emptyset$.

\end{enumerate}	
\end{theorem}
\smallskip
\begin{proof}
{\bf(A)} Let $\rho \in \D_+(\H)$. Then
$$
{\rm tr}(|0\>\<0|\rho)=\frac{1}{\pi} \int_{\RE^2} e^{-z^2} W\rho (z) dz =0
$$
since ${\rm Ran}\, \rho \subset \H \perp |0\>$. This, however, is not possible if $W\rho(z) \ge 0$, $\forall z\in \RE^2$. Hence, $\D_+(\H) =\emptyset$.

\medskip

{\bf(B)} If $|0\rangle  \in \H$ then $\rho_0= |0\rangle \langle 0|  \in \D_+(\H)$ and $\N(\rho_0)=\emptyset$. 

In the trivial case $\H=\H^0$, it follows from \cref{ipos} that $\rho_0 \in {\rm int}_{\A(\H) } \D_+(\H)$, and thus ${\rm int}_{\A(\H) } \D_+(\H)$ is nonempty. 

We then consider the case ${\rm dim}\, \H \ge 2$. Let  $\rho_1 \in \D(\H)$ be of full rank. Then $\rho_1\not=\rho_0$ and all states $\eta \in (\rho_0,\rho_1]$ are also of maximal rank (because they belong to the intrinsic interior of $\D(\H)$, cf.~Lemma \ref{Simon}-B and \cref{intQS}-A). We then have two possibilities. If $\rho_1 \in \D(\H) \backslash \D_+(\H)$ then from Lemma \ref{rho+}-A there exists $\rho \in (\rho_0,\rho_1) \cap \D_{0+}(\H)$ such that, for all $\eta \in [\rho_0,\rho)$, we have $\N(\eta)=\emptyset$. If $\rho_1 \in \D_+(\H)$ then all $\eta \in (\rho_0,\rho_1)$ satisfy $\N(\eta)=\emptyset$ (which follows trivially from $\N(\rho_0)=\emptyset$). Hence, in both cases, we conclude from Theorem \ref{ipos} that the states $\eta$ above belong to ${\rm int}_{\A(\H)} \D_+(\H)$, and so ${\rm int}_{\A(\H)} \D_+(\H)$ is nonempty. 

Finally, from Lemma \ref{Simon}-A we conclude that ${\rm Aff}(\D_+(\H))=\A(\H)$.

\medskip

{\bf(C)} Consider the Hilbert spaces $\H_1=$ span$\{|1\rangle,\tfrac{\sqrt{10}}{10}(3|0\rangle+|2\rangle)\}$ and 
$\H_2=$ span$\{|3\rangle,\tfrac{\sqrt{10}}{10}(|0\rangle-3|2\rangle)\}$. Since $\H_1 \perp \H_2$, we also have $\D(\H_1) \perp \D(\H_2)$ and so it is not possible for both sets $\D_+(\H_1)$ and $\D_+(\H_2)$ to be nonempty. 

We can show numerically that there exists $\xi \in\D_+(\H_1)$ with empty nodal set (the diagonal elements of $\D_+(\H_1)$ correspond to the hypotenuse of the triangle in \cref{fig:combined}(right) with $c=3$, see below). Following the steps of (B) above (but using $\xi$ instead of $\rho_0$, and choosing $\rho_1 \not=\xi$) we conclude that ${\rm int}_{\A(\H_1)}\D_+(\H_1) \not=\emptyset$.   
\end{proof}

\subsection{Topology of Wigner-positive states in infinite dimension}

In this section, we turn to the infinite-dimensional case.

\begin{theorem}\label{Compact_infinite_dim}
The set $\D_+$ is convex, bounded and closed, but not compact.
\end{theorem}

\begin{proof} 
	
It is obvious that $\D_+$ is convex and bounded. The proof that $\D_+$ is closed is identical to the finite-dimensional case (see \cref{Compact_finite_dim}). Since $\B_1$ is a Banach space, this implies that $\D_+$ is a complete metric space.

To prove that $\D_+$ is not compact, we let $a>0$ and consider the sequence of coherent states, where $D(z)$ is the displacement operator given in Eq.~(\ref{D(z)}):
\begin{equation}
    \rho_n(a)\coloneqq D(na,0)\ket{0}\!\bra{0}(D(na,0))^{\dagger},
\end{equation}
for all $n\in\mathbb N$. For all $m\neq n$ we have
\begin{equation}
    \|\rho_m(a)-\rho_n(a)\|_1=2\sqrt{1-e^{-a^2(m-n)^2}}\ge2\sqrt{1-e^{-a^2}},
\end{equation}
so $(\rho_n(a))_{n\in\mathbb N}$ is a sequence of points in $\D_+$ with no Cauchy subsequence and therefore no converging subsequence, so $\D_+$ is not compact since it is complete.
\end{proof}

When $\H$ is infinite-dimensional, the following result is a straightforward consequence of \cref{intQS}:

\begin{corollary} \label{intW_+}We have
$$
{\rm int}_{\A} \D_+ \subseteq {\rm iint}\, \D =\emptyset.
$$
\end{corollary}

We now prove a stronger result:

\begin{theorem} \label{intW_+2}
The set $\D_+$ has empty relative interior with respect to $\D$ for the trace norm, i.e.~${\rm int}_{\D} \D_+ = \emptyset$.
	
\end{theorem}

\begin{proof}
Let $\rho\in\D_+$. For each $n\in \IN$, let $z_n\in \RE^2$ be such that $W\rho(z_n)<\frac1{2\pi n}$. Note that such $z_n$ exists for all $n\in \IN$; otherwise $W\rho(z)\ge\frac1{2\pi n}$ for some $n$ and all $z\in \RE^2$, and thus $W\rho \notin L^2(\RE^2)$. Define
\begin{equation}\label{sigma_n}
    \sigma_n\coloneqq\left(1-\frac1n\right)\rho+\frac1n D(z_n)|1\rangle\langle1|(D(z_n))^{\dagger},
\end{equation}
for all $n>0$. By construction, $W{\sigma_n}(z_n)\le-\frac1{2\pi n}$ (note that $Wh_1(0)=-\tfrac{1}{\pi}$). Hence, $\sigma_n\in\D\setminus\D_+$ and $\sigma_n\rightarrow\rho$ in the trace norm when $n\rightarrow+\infty$.
\end{proof}

To proceed we note that \cref{rho+} extends trivially to the  infinite-dimensional case: 

\begin{corollary}\label{CorollaryX}
Let $\rho_0 \in \D_+$ be such that $\N(\rho_0)=\emptyset$\footnote{We do not assume that $\rho_0=|0\>\<0|$.}, and let $\rho_1\in \D \backslash \D_+$ have a bounded nodal set $\N(\rho_1)$. Then:
 $$
 \exists^1 \rho \in (\rho_0, \rho_1) \, \cap \, \D_{0+}
 $$
and $\|\rho-\rho_0 \|_1 < \|\rho_1-\rho_0 \|_1$.
\end{corollary}

This result is used in the proof of the next theorem:

\begin{theorem}\label{X}
In the infinite-dimensional case $\H=L^2(\RE)$, we have:
\begin{equation}\label{DenseD0+}
\D_+=\overline{\D_{0+}}
\end{equation}
where the closure is with respect to the trace norm.
\end{theorem}

\begin{proof}
Since $\D_+$ is closed, and $\D_{0+} \subset \D_+$, we have:
\begin{equation*}
\overline{\D_{0+}} \subseteq \D_+.
\end{equation*}
Let us prove the converse inclusion. Let $\rho \in \D_+$  and assume that $\rho \notin \D_{0+}$ (otherwise we already have $\rho \in \overline{\D_{0+}}$). Consider the sequence $(\sigma_n)_{n\in \IN} \subset \D\backslash \D_+ $ defined by Eq.~(\ref{sigma_n}). We have $\sigma_n \to \rho$ in the trace norm. 

Since 
$$
\Omega_-(\sigma_n) \subseteq \Omega_-(D(z_n)|1\rangle\langle1|(D(z_n))^{\dagger})\, , \quad \forall n\in \IN
$$  
we conclude that $\N(\sigma_n)$ is bounded for all $n\in \IN$. 
Hence, from \cref{CorollaryX} there exists $\eta_n\in (\rho,\sigma_n)\cap \D_{0+}$ which also satisfies
$$
\|\eta_n-\rho \|_1 < \|\sigma_n-\rho \|_1.
$$ 
Since $\sigma_n \to \rho$ in the trace norm, the sequence $(\eta_n)_{n\in \IN} \subset \D_{0+}$ also satisfies $\eta_n \xrightarrow[]{\| \, \|_1} \rho$ and thus $\D_+ \subseteq \overline{D_{0+}}$. 
\end{proof} 

Note from Theorems \ref{ipos} and \ref{intWN}-B that Eq.~(\ref{DenseD0+}) does not hold in the finite-dimensional case $\H \subseteq \H^N$. 
In the next section we will see that, in this case, we have instead: $\D_+(\H)= \overline{\rm conv} (\D_{0+}(\H))$. 

\section{Geometry}
\label{sec:geom}

Having characterised the topological properties of the sets of WPS, we turn to their geometrical properties.

\subsection{Geometry of Wigner-positive states in finite dimension}
\label{sec:geom_finite}

In this section we characterise the geometric structure of the convex sets $\D_+(\H)$ in the finite-dimensional case. 

\subsubsection{Characterisation via the extreme points}

One of the main problems we wish to address is that of identifying and constructing a set of WPS $Z\subset \D_+(\H)$ such that $\D_+(\H)=$ $\overline{\mbox{\rm conv}}$$(Z)$. Since $\D_+(\H)$ is convex, a natural candidate is the set of its extreme points. In this section we consider this problem in the finite-dimensional case $\H\subseteq \H^N$. 

Since the set $\D_+(\H)$ is compact and convex, the Krein--Milman Theorem readily applies (cf.~Theorem \ref{Krein--Milman} and Eq.~(\ref{KM-eq-finite})):

\begin{lemma} \label{KMfinite}
Let $\H\subseteq \H^N$, $N\in \IN_0$. If not empty, the set $\D_+(\H)$ is equal to the convex hull of its extreme points.
\end{lemma}

The main problem is then how to identify the extreme points. The next theorem provides a partial answer:

\begin{theorem} \label{Extreme_points}
Let $\H\subseteq \H^N$. We have two slightly different cases:
\begin{itemize}
\item If $|0\>\in \H$ then $\E(\D_+(\H)) \subseteq \D_{0+}(\H) \cup \{|0\>\<0|\}$.
\smallskip
\item If $|0\>\notin \H$ then $\E(\D_+(\H)) \subseteq \D_{0+}(\H) $.
\end{itemize}
\end{theorem}

\begin{proof}
If $|0\>\in \H$ then, by Hudson's theorem, $\rho_0=|0\>\<0|$ is the unique pure state in $\D_+(\H)$, and $\rho_0 \in \E(\D_+(\H))$. If $|0\>\notin \H$ then
there are no pure states in $\D_+(\H)$. Hence, the proof of the theorem reduces (in both cases) to prove that if $\rho \in \E(\D_+(\H))$ is a mixed state then $\N(\rho) \not=\emptyset$.

Since $\rho$ is a mixed state, $\rho \in (\rho_0,\rho_1)$ for some $\rho_0,\rho_1 \in \D(\H) \backslash \{\rho \}$ and, at least, $\rho_1 \notin  \D_+(\H)$ (otherwise $\rho$ would not be an extreme point of $\D_+(\H)$).  

Let us then define: 
$$
\rho_t\coloneqq (1-t) \rho_0 + t \rho_1 \quad , \quad t \in [0,1]
$$
so that $\E(\D_+(\H)) \ni \rho=\rho_{t_0}$ for some $t_0\in (0,1)$.
Let also $F_t(z)\coloneqq W\rho_t(z)$.

Since $\rho_t\in \D(\H)$ and $\H \subseteq \H^N$, the sets $\Omega_-(\rho_t)=\{z\in \RE^2|\, F_t(z) \le 0\}$, if nonempty, are all compact (cf.~\cref{Theorem5} and \cref{SecondTheorem}).

Now note that we can have two seemingly different cases: $\rho_0\in \D_+(\H)\backslash\{\rho\}$ and $\rho_0 \in \D(\H) \backslash \D_+(\H) $. Let us consider first the case $\rho_0\in \D_+(\H)\backslash\{\rho\}$.   

Then, for every $z\in \Omega_- (\rho_1)$, the function $t \to F_t(z)$, $t\in (t_0,1]$ is decreasing (possibly constant if $F_1(z)=0$) and so the compact sets $\Omega_- (\rho_t)$ satisfy:
$$
t_0 < t \le  t' \le 1 \, \Longrightarrow \, \Omega_- (\rho_t) \subseteq \Omega_- (\rho_{t'}) ~.   
$$

Moreover, the sets $\Omega_- (\rho_t)$, $t_0<t\le 1$ are all nonempty (otherwise $\rho_t \in \D_+(\H)$ for some $t>t_0$ and $\rho_{t_0}\in (\rho_0,\rho_t)$, with $\rho_0,\rho_t \in \D_+(\H)$, could not be an extreme point of $\D_+(\H)$). Hence, $\left(\Omega_- (\rho_t)\right)_{t_0 < t \le 1}$ is a nested collection of nonempty compact sets, and thus has a nonempty intersection (Cantor's intersection theorem - adapted to the case of a continuous one-parameter family of nested sets). Let then:
$$
z_0 \in \bigcap_{t_0 < t \le 1} \Omega_- (\rho_t).
$$   
Since $F_t(z_0) \le 0$ for all $t \in (t_0,1]$, we have:
$$
\lim_{t\to t_0^+} F_t(z_0) \le 0
$$
and since $t \to F_t(z_0)$ is continuous and $F_{t_0}$ is non-negative, we must have $W\rho(z_0)=F_{t_0}(z_0)=0$, and so $\N(\rho) \not=\emptyset $.

We now consider the case $\rho_0 \in \D(\H) \setminus \D_+(\H)$. Note that $\Omega_-(\rho_0) \cap \Omega_-(\rho_1) = \N(\rho_0) \cap \N(\rho_1)$ (otherwise $\rho_t \notin \D_+(\H) $ for all $t\in [0,1]$). Moreover, if 
$\N(\rho_0) \cap \N(\rho_1) \not= \emptyset$ then $\N(\rho_t) \not=\emptyset$ for all $t\in [0,1]$, and so $\rho=\rho_{t_0} \in \D_{0+}(\H)$. 

Hence, assume that $\Omega_-(\rho_0) \cap \Omega_-(\rho_1) = \emptyset $. For every $z\in \Omega_- (\rho_1)$ we consider again the function $t \to F_t(z)$, $t\in (t_0,1]$ which is strictly decreasing. In this case we also consider, for every $z\in \Omega_- (\rho_0)$, the function $t \to F_t(z)$, $t\in [0,t_0)$, which is strictly increasing. It follows that the sets $\Omega_-(\rho_t)$, $t \in [0,1] \backslash \{t_0\}$ satisfy:
$$
0 \le t \le  t' < t_0 \, \Longrightarrow \, \Omega_- (\rho_t) \supseteq \Omega_- (\rho_{t'}) ,   
$$  
$$
t_0 < t \le  t' \le 1 \, \Longrightarrow \, \Omega_- (\rho_t) \subseteq \Omega_- (\rho_{t'}) .   
$$  
Moreover, all the sets in one of these two families are nonempty (otherwise $\rho_t \in \D_+(\H)$ for some $t_1 <t_0$ and some $t_2 > t_0$, and so $\rho_{t_0} \in (\rho_{t_1},\rho_{t_2})$ could not be an extreme point of $\D_+(\H)$).

We conclude that either $\left(\Omega_- (\rho_t)\right)_{0 \le  t < t_0}$ or $\left(\Omega_- (\rho_t)\right)_{t_0 < t \le 1}$ is a nested collection of nonempty compact sets, and the rest of the proof follows the steps of the previous case.   
\end{proof}

\begin{remark}
	
If the negative set $\Omega_-(\rho_1)$ is unbounded, the previous result is not valid, as the transition from non-positive to strictly positive Wigner states does not necessarily cross a vanishing WPS. At least, this is the case for $L^2(\RE)$-functions, as the following example shows:
	
	Let $f_0(x)=e^{-x^2}$ and $f_1(x)=-\frac{2}{\pi} e^{-x^2} \arctan |x|$. Then $f_t \in [f_0,f_1]$ is of the form:
	$$
	f_t(x)=(1-t)f_0(x)+tf_1(x)=e^{-x^2}-te^{-x^2}\left( 1+\frac{2}{\pi} \arctan |x| \right)\, , \, t\in [0,1]
	$$
	and 
	$$
	f_t(x) \ge 0\, ,\, \forall x\in \RE \, \Longleftrightarrow \, t\le 1/2~.
	$$
However
$$
f_t(x) =0 \Longleftrightarrow \arctan |x| = \frac{\pi}{2} \left( \frac{1}{t} -1\right). 
$$ 
Observe that if $0< t \le 1/2$, the previous equation has no solutions. We conclude that if $f$ is non-negative then $f$ does not vanish; and thus, in this case, there is no vanishing positive function. 

\end{remark}

We then obtain the following simple characterisation of the finite-dimensional sets of WPS.

\begin{corollary}\label{GeneratorsVWPS}
Let $\H\subseteq \H^N$ for some $N\in \IN_0$ and let $\rho_0=|0\>\<0|$. We have:
\begin{itemize}

\item If $|0\> \in \H$ then $
\D_+(\H) = {\rm conv} \left(\D_{0+}(\H) \, \cup \, \{\rho_0\} \right)
$
\smallskip
\item If $|0\> \notin \H$ then $
\D_+(\H) = {\rm conv} (\D_{0+}(\H) )
$ 
\smallskip
\item If ${\rm dim}\,\H \ge 2$ then  $\D_+(\H) = \overline{\rm conv} (\D_{0+}(\H) )$
\end{itemize}	
\end{corollary}
\smallskip
\begin{proof}
The two first results follow directly from Lemma \ref{KMfinite}, Theorem \ref{Extreme_points}, and the fact that if $\E(A)\subseteq B \subseteq$ ${\rm conv}(\E(A))$ then ${\rm conv}(B)= {\rm conv}(\E(A))$.

Let us prove the third result. If $|0\> \notin \H$ the relation follows from the second identity (taking into account that $\D_+(\H)$ is closed). 

In the case $|0\> \in \H$, let us first show that $\rho_0 \in  \overline{\D_{0+}(\H)}$. Since dim $\H\ge 2$, we select a state $|\phi\> \in \H$ such that $|\phi\> \perp |0\>$, and construct the  sequence of pure states:
$$
|\psi_n\>=\cos\left(\tfrac{1}{n}\right) |0\> + \sin\left(\tfrac{1}{n}\right) |\phi\>\, , \quad n\in \IN.
$$ 
Let $\rho_n=|\psi_n\>\<\psi_n|$. The only Gaussian pure state in $\D(\H^N)$ is $\rho_0$ and thus, from Hudson's Theorem, $\rho_n\in \D(\H) \backslash \D_+(\H)$ for all $n \in \IN$. From Lemma \ref{rho+} there exists $\wt\rho_n \in (\rho_0,\rho_n)\cap \D_{0+}(\H)$ that satisfies:
$$
\| \wt\rho_n-\rho_0\|_1 \le \|\rho_n-\rho_0\|_1.
$$
We also have: 
$$
\|\rho_n-\rho_0\|_1=2\sqrt{1- |\<\psi_n|0\>|^2}=2\sin\left(\tfrac{1}{n}\right).
$$
Hence, $(\wt\rho_n)_{n\in \IN} \subset \D_{0+}(\H)$ and $\wt\rho_n \xrightarrow[n\to \infty]{\| \, \|_1} \rho_0$. We conclude that $\rho_0\in \overline{D_{0+}(\H)}$. 

Now recall that, in finite dimensions, if $X$ is bounded then $\overline{{\rm conv}}(X)={\rm conv}(\overline{X})$ \cite[Theorem 17.2]{Rockafellar}. Taking into account that $\D_+(\H)$ is compact, we obtain:
$$
\D_+(\H) ={\rm conv} \left(\D_{0+}(\H) \, \cup \, \{\rho_0\} \right)= \overline{{\rm conv}} \left(\D_{0+}(\H) \, \cup \, \{\rho_0\} \right)
$$
$$
={\rm conv} \left(\overline{\D_{0+}(\H) \, \cup \, \{\rho_0\}} \right)={\rm conv} \left(\overline{\D_{0+}(\H)} \right)
= \overline{{\rm conv}} \left(\D_{0+}(\H) \right)
$$
which concludes the proof. 
\end{proof}

\medskip

Note that in general $\E(\D_+(\H))$ is a proper subset of $ \D_{0+}(\H) \cup \{\rho_0\}$, and so it may be possible to fine-tune the set of generators proposed in Corollary \ref{GeneratorsVWPS}. We discuss this problem by further exploring the very interesting relation between extreme points and nodal sets.  

Consider the {\it faces} of $\D_+(\H)$:
$$
\F(S)\coloneqq\left\{\rho \in \D_+(\H): \, S \subseteq \N(\rho)  \right\} \subseteq \D_{0+}(\H)
$$
associated to the (nonempty) sets $S\subset \RE^2$. These are convex sets, and satisfy
$\F(S) \subseteq \E(\D_+(\H))$ if and only if $\F(S)$ is a singleton. Moreover, up to the possible exclusion of $\{\rho_0\}$, we have:
\begin{equation}\label{ExtFaces2}
\D_{0+}(\H)=\bigcup_{S\subset \RE^2}\F(S) \quad \mbox{and} \quad \E(\D_+(\H))=
\bigcup_{S\subset \RE^2}\E(\F(S))
\end{equation}
where the second equality follows from (cf.~Eq.~(\ref{ExtFaces})):
$$
\E(\F(S))= \E(\D_+(\H)) \, \cap \, \F(S)
$$
by taking the union $\cup_{S\subset \RE^2}$ (and ignoring $\{\rho_0\}$):
\begin{eqnarray}\label{ExtrF}
\bigcup_{S\subset \RE^2}\E(\F(S))&=&\bigcup_{S\subset \RE^2}\big( \E(\D_+(\H)) \, \cap \, \F(S)\big)\\
&=&\E(\D_+(\H)) \cap \big(\bigcup_{S\subset \RE^2} \F(S) \big)\nonumber\\ 
&=&\E(\D_+(\H)) \cap \D_{0+}(\H)\nonumber\\ 
&=&\E(\D_+(\H))\nonumber
\end{eqnarray}
where we used the first identity in Eq.~(\ref{ExtFaces2}) and Theorem \ref{Extreme_points}.
 We thus conclude that:
\begin{corollary}\label{Faces_Cor}
For $\H\subseteq \H^N$ we have
$$
\E\left(\D_+(\H)\right)=  \bigcup_{S\subset \RE^2}\E(\F(S)) \, \cup \, \{\rho_0 \} 
$$
or 
$$  
\E\left(\D_+(\H)\right)=  \bigcup_{S\subset \RE^2}\E(\F(S)) 
$$
depending on whether $|0\>\in \H $ or $|0\>\notin \H$. 
\end{corollary}

An interesting complement to the previous results, is the following:

\begin{theorem}
Let $\H=\H^N_{z_0,S}$. The extreme points of $\D_+(\H)$ and $\D_+^{N}$ are related by:
$$
\rho \in \E(\D^N_+) \Longleftrightarrow U \rho \, U^{\dagger} \in \E(\D(\H)) 
$$
where $U=D(z_0) \mu(S^{-1})$. 	
\end{theorem}
\begin{proof}
Let $\eta \in \D(\H)$. Then exists $\rho \in \D^N$ such that $\eta= U\rho \,U^{\dagger}$ (cf.~Eq.~(\ref{Hilbert_zS})). Moreover $\eta$ is a WPS if and only if $\rho$ is a WPS. This follows from 
$$
W\eta(z)=W\rho(S(z-z_0)).
$$  
Assume that $\rho \notin \E(\D^N_+)$. Then exists $\rho_0,\rho_1\in \D^N_+$ such that $\rho \in (\rho_0,\rho_1)$. This implies that $\eta =U\rho \, U^{\dagger} \in (\eta_0,\eta_1)$ where $\eta_i=U\rho_i \, U^{\dagger}$, $i=0,1$. Hence, $\eta \notin \E(\D_+(\H))$.

Since $\eta =U\rho \, U^{\dagger} \Longrightarrow \rho =  U^{\dagger} \eta \, U$, the converse implication also holds, concluding the proof. 
\end{proof}

\subsubsection{Characterisation via the boundary}

We now concentrate on the case $|0\> \in \H$, and provide an alternative construction of $\D_+(\H)$ by using convex combinations of certain elements of the boundary of $\D_+(\H)$. The main advantage here is that we are able to construct the set of generators explicitly by using convex combinations of $\rho_0=|0\>\<0|$ and particular density matrices from the boundary of $\D(\H)$. As a side result we also show that the entire set of density matrices $\D(\H)$ can be construct by taking affine combinations of the same set of generators.

Let $\H \subseteq \H^M$ be such that $|0\> \in \H$ and dim $\H =N+1$ with $0< N \le M$. Define the set: 
\begin{equation}\label{Y}
    Y(\H)= \{\rho \in \D(\H): {\rm Rank}\, \rho =N \, \wedge \, |0\rangle \notin {\rm Ran}\, \rho \}
\end{equation}
which is a subset of $\partial^i \D(\H)$. We then have:

\begin{theorem}\label{X_+}

Let $\rho_0 =|0\rangle\langle 0|$ and let $\rho \in {\rm iint} \,\D_+(\H)$. Then:

\begin{enumerate}

\item[{\bf(A)}] $\exists^1 \, {\rho_1 \in Y(\H)}: ~\rho \in (\rho_0, \rho_1)$.
\smallskip
\item[{\bf(B)}] $\exists^1\, {\rho_+ \in \partial^i \D_+(\H)}:~ \rho \in (\rho_0, \rho_+)$. 
 \smallskip 
\item[{\bf(C)}] $\rho_+=F(\rho_1)$, where $F:Y(\H) \longrightarrow \partial^i\D_+(\H)$ is defined by
\begin{equation}\label{F1}
F(\rho_1)= (1-t_0(\rho_1)) \rho_0 + t_0(\rho_1)\rho_1
\end{equation}
where
\begin{equation}\label{F2}
t_0(\rho_1) =\frac{1}{1-k_0(\rho_1)}, \quad k_0(\rho_1)= {\rm min}\left\{0, \, \tfrac{W\rho_1(z)}{W\rho_0(z)} \, , \, z\in \RE^2\right\}
\end{equation}
Moreover, $F$ is injective, and $t_0(\rho_1)=\frac{||\rho_+-\rho_0||_1}{||\rho_1-\rho_0||_1}$.

\end{enumerate}

\end{theorem}

\begin{figure}[t]
  \centering
  \includegraphics[width=1\textwidth]{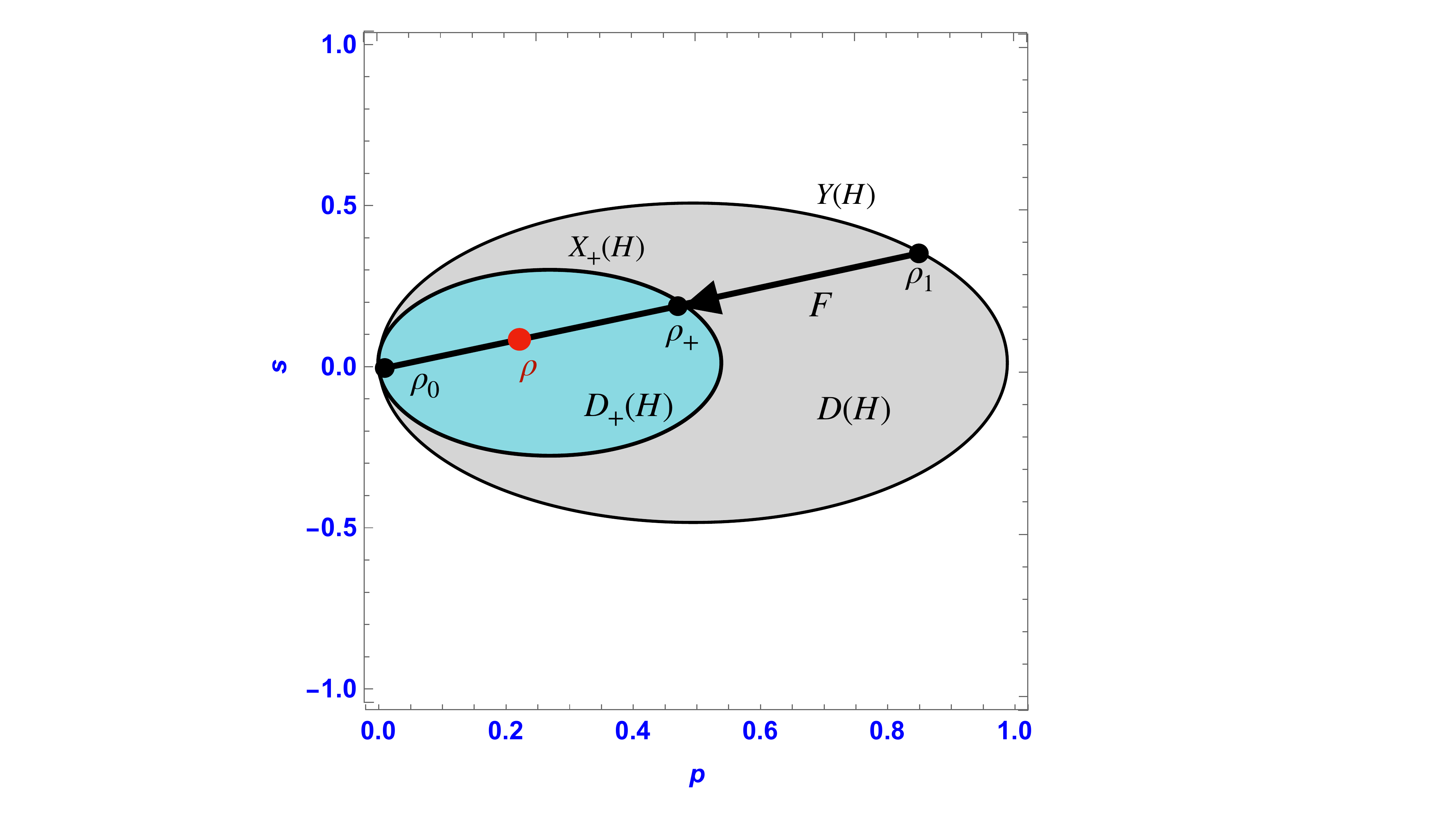}
  \caption{\small Illustration of \cref{X_+} for the case $\H=\H^1={\rm span}\{|0\>,|1\>\}$.  Each point $(p,s)$ inside the larger ellipse  represents a density matrix  
$
\rho{(p,s)}= (1-p) \,\rho_0+ p \,|1\>\<1| +s [|0\>\<1| + |1\>\<0| ],
$ (cf.~\cref{DM} below).
The set of density matrices $\D(\H)$ (grey region) and of WPS $\D_+(\H)$ (blue region) are represented for the case of $s$ real. The origin corresponds to the state $\rho_0=|0\>\<0|$. In this case, $Y(\H)=\partial^i\D(\H) \backslash \{\rho_0\}$ and $X_+(\H)\coloneqq{\rm Ran} \, F=\partial^i\D_+(\H) \backslash \{\rho_0\}$.}
  \label{fig:Y}
\end{figure}

\begin{proof}
{\bf(A)} and {\bf(B)}: The first part of the proof follows from Lemma \ref{Simon}-D. Since $\D(\H)$ and $\D_+(\H)$ are both compact and $\A(\H)={\rm Aff}(\D(\H))= {\rm Aff}(\D_+(\H))$ (cf. Theorem \ref{intWN}-A), we have from Lemma \ref{Simon}-D that for every $\rho \in {\rm iint}\, \D_+(\H) \subset {\rm iint}\, \D(\H)$ statement (B) holds and statement (A) holds for some $\rho_1 \in \partial^i \D(\H)$.

We then prove that $\rho_1$ actually belongs to the subset $Y(\H) \subset \partial^i \D(\H)$. Since $\rho \in (\rho_0,\rho_1)$ and $\rho_1 \in \partial^i \D(\H)$, we must have Rank $\rho_1=N$ (note that $\rho \in {\rm iint} \,\D_+(\H) \Longrightarrow$ Rank $\rho=N+1$). Moreover 
$$
{\rm Ran~} \rho={\rm span}\{|0\rangle,{\rm Ran~} \rho_1 \}=\H  
$$
and so $|0\rangle \notin $ Ran $\rho_1$. We conclude that $\rho_1 \in Y(\H)$. 

\smallskip

{\bf(C)} From (A) and (B) we have $\rho \in (\rho_0,\rho_1)$ and $\rho \in (\rho_0,\rho_+)$. Hence, $\rho_+\in (\rho_0,\rho_1]$.

If $\rho_1\in \D_+(\H)$ (and since $\rho_1\in Y(\H) \subset \partial^i\D(\H)$) then $\rho_1\in \partial^i \D_+(\H)$ (cf.~Corollary \ref{boundary}-B). It follows from (B) above that $\rho_+=\rho_1$. This corresponds to the case $k_0(\rho_1)=0$ in Eqs.~(\ref{F1},\ref{F2}).

If $\rho_1\notin \D_+(\H)$, we know from Lemma \ref{rho+} how to calculate $\rho_+$ from $\rho_1$. Since $\rho_+\in \partial^i\D_+(\H)$ is of maximal rank, we also have $\rho_+\in \D_{0+}(\H)$ (cf.~Corollary \ref{boundary}-A), and so $\rho_+$ is the unique state in $(\rho_0,\rho_1)\cap \D_{0+}(\H)$ (cf.~Lemma \ref{rho+}). This state is explicitly given by Eq.~(\ref{eqrho+}). Hence, we have $\rho_+=F(\rho_1)$ for $F$ given by Eqs.~(\ref{F1},\ref{F2}).

Moreover, if $F$ is not injective then $\rho_+ \in [\rho_0, \rho_1]$ and $\rho_+ \in [\rho_0, \rho_1']$ for two different states $\rho_1,\rho_1' \in \partial^i \D(\H)$, which is not possible by Lemma \ref{Simon}-D.

Finally, we show that $F$ is of the form $F:Y(\H) \to \partial^i\D_+(\H)$, i.e.~that Ran $F \subseteq \partial^i\D_+(\H)$. We have two cases: 1) if $\rho_1\in Y(\H) \cap \D_+(\H)$ then $\rho_1 \in \partial^i\D_+(\H)$ and $F(\rho_1)=\rho_1 \in \partial^i\D_+(\H)$. 2) If $\rho_1\in Y(\H) \setminus \D_+(\H)$ then it follows from Eq.~(\ref{eqrho+}) that $F(\rho_1) \in \D_{0+}(\H) \subset \partial^i\D_+(\H)$.    
\end{proof}

\smallskip

We can now construct the WPS as {\it  convex combinations} of $\rho_0=|0\rangle\langle 0|$ and some specific states $\rho_1\in \partial^i\D_+(\H)$ (Corollaries \ref{C18} and \ref{C18-B}), and the non-positive Wigner states as {\it affine combinations} of the same states (\cref{C19}). 

\begin{corollary}\label{C18}
Let $X_+(\H)= {\rm Ran}\, F$, where $F$ was given by Eqs.~(\ref{F1},\ref{F2}). Then:
\begin{equation}\label{C18-eq}
{\rm iint}\, \D_+(\H)=\bigcup_{\rho_1\in Y(\H)} (\rho_0,F(\rho_1))=\bigcup_{\rho_+\in X_+(\H)} (\rho_0,\rho_+)
\end{equation}
These descriptions of ${\rm iint}\, \D_+(\H)$ are optimal in the sense that: 
\begin{enumerate}
	
\item[(i)] every interval is nonempty and 
\smallskip

\item[(ii)]  each element of ${\rm iint}\,\D_+(\H)$ belongs to one, and only one, of the intervals.   
\end{enumerate}
Moreover, Eq.~(\ref{C18-eq}) implies:
\begin{equation}\label{C18-eq2}
\D_+(\H)=\ol{\bigcup_{\rho_+\in X_+(\H)} (\rho_0,\rho_+)}
\end{equation}
\end{corollary}

\begin{proof}
It is obvious from \cref{X_+} that 
$$
{\rm iint} \,\D_+(\H) \subseteq \bigcup_{\rho_1\in Y(\H)} (\rho_0,F(\rho_1)).
$$
Let us prove the converse inclusion. We also know from \cref{X_+} that for an arbitrary $\rho_1 \in Y(\H)$ we have $F(\rho_1) \in \D_+(\H)$. Moreover, $F(\rho_1) \not=\rho_0$ (because $F(\rho_1)=\rho_1\not=\rho_0$ if $\rho_1\in Y(\H)\cap \D_+(\H)$, and $F(\rho_1) \in \D_{0+}(\H)$ if $\rho_1 \in Y(\H)\setminus\D_+(\H)$). 

Hence, $F(\rho_1) \in \D_+(\H) \setminus \{\rho_0\}$, and thus the interval $(\rho_0,F(\rho_1))$ is always nonempty, which proves statement (i) above. 

Moreover, it follows from   
$\N(\rho_0)=\emptyset $ and  $\F(\rho_1) \in \D_+(\H)$ that every $\rho \in (\rho_0,F(\rho_1))$ is a WPS with empty nodal set.

Since $(\rho_0,F(\rho_1)) \subseteq (\rho_0,\rho_1)$, whose elements are all of full rank, we conclude from \cref{ipos} that every $\rho \in (\rho_0,F(\rho_1))$ belongs to ${\rm iint} \,\D_+(\H)$, which proves the converse inclusion. Hence, Eq.~(\ref{C18-eq}) holds.

Again, \cref{X_+}-B implies that every $\rho \in {\rm iint} \,\D_+(\H)$ belongs to just one of the intervals $(\rho_0,\rho_+)$, $\rho_+\in X_+(\H)$. Since $F$ is injective, $\rho_+=F(\rho_1)$ also for just one $\rho_1\in Y(\H)$. This proves (ii).

Finally, Eq.~(\ref{C18-eq2}) follows from the fact that $\D_+(\H)$ is closed. 
\end{proof}

Quite interestingly, the same generators can be used to construct the set of Wigner non-positive states. This is related to the fact that Aff$(\D_+(\H))=$ Aff$(\D(\H))$. 

\begin{corollary}\label{C19}
Let $\rho \in \D(\H)\backslash \D_+(\H)$. Then there exists $\rho_+\in \partial^i \D_+(\H) $ such that:
$$
\rho= (1+s) \rho_+-s\rho_0 \quad \mbox{for some} \quad s> 0.
$$
If $\rho \in {\rm iint}\D(\H)\setminus \D_+(\H)$ then $\rho_+\in X_+(\H)$.
\end{corollary}

\begin{proof}
From Lemma \ref{rho+} we know that for every $\rho\in \D(\H)\backslash \D_+(\H)$ there exists $\rho_+\in \D_{0+}(\H) \subset \partial^i\D_+(\H)$, and $0<t_0 <1$ such that:
\begin{equation}\label{rho++}
\rho_+=(1-t_0)\rho_0+t_0\rho \quad \Longleftrightarrow \quad \rho=(1+s)\rho_+-s\rho_0 \, , \quad s > 0
\end{equation}
where we defined $s=\frac{1}{t_0} -1$, which satisfies $s>0$.

We now consider the case $\rho \in {\rm iint}\D(\H)\setminus \D_+(\H)$. Following exactly the steps of the proof of \cref{X_+}-A, we conclude that:
$$
\exists^1 {\rho_1\in Y(\H)}: \, \rho \in (\rho_0,\rho_1).
$$
Consider the state $\rho_+$ given in Eq.~(\ref{rho++}). It follows from \cref{rho+} that $\rho_+$ is the unique state in $(\rho_0,\rho_1) \cap \D_{0+}(\H)$. Moreover, \cref{rho+}-A and -B also implies that $(\rho_0,\rho_1) \cap \partial^i \D_{+}(\H)=\{\rho_+\}$ (note that $\rho_+$ is of full rank). Since $F(\rho_1) \in (\rho_0,\rho_1) \cap \partial^i \D_{+}(\H)$ (cf. \cref{X_+}), we necessarily have $\rho_+=F(\rho_1) \in X_+(\H)$, concluding the proof.
\end{proof}

\medskip

Finally, let us make a few remarks about the sets of generators $Y(\H)$ and $X_+(\H)$. In the next subsection we will construct these sets explicitly for some simple examples.

Let us introduce some notation. For $v\in \H$, define: 
$$
\<v\>:= {\rm span}\,\{|v\>\} \quad  , \quad \<v\>^\perp:= \{u\in \H: \<u|v\>=0\}
$$ 
so that $\<v\> \oplus \<v\>^\perp=\H$.  

Since every $\rho_1 \in Y(\H) \subset \partial^i\D(\H)$ is of rank $N$ and dim $\H=N+1$ (cf.~Eq.~(\ref{Y})), there exists $v \in \H$ such that $\rho_1 \in \D(\<v\>^\perp)$. It follows that:
\begin{eqnarray}\label{Yold}
Y(\H) &= & \bigcup_{v\in \H:\, \<0|v\> \not=0} \{\rho_1 \in \D( \<v\>^\perp) : \, \rho_1 \mbox{ is of full rank}\} \nonumber \\
&=& \bigcup_{\<v\> \in {\mathbb P}(\H) \wedge \<0|v\> \not=0} {\rm iint} \,  \D(\<v\>^\perp) 
\end{eqnarray}
where ${\mathbb P}(\H)$ is the projective space of lines in $\H$. Note also that every $\rho\in Y(\H)$ belongs to one and only one, of the sets $\D(\<v\>^\perp)$. We conclude from Eq.~(\ref{C18-eq}) that:

\begin{corollary}\label{C18-B} Let $\H\subseteq \H^M$, $M\ge 1$ and $|0\> \in \H$. Then: 
$$
 \D_+(\H)= \overline{\bigcup_{\<v\> \in {\mathbb P}(\H) \wedge \<0|v\> \not=0} {\rm conv}\big(\left\{\rho_0 \right\} \cup F\big(\D(\<v\>^\perp)\big)\big)}.
$$
This formulation partitions $\D_+(\H)$ into affine cones ${\rm conv}\big(\left\{\rho_0 \right\} \cup F\big(\D(\<v\>^\perp)\big)\big)$, which are parametrized by the directions $\<v\> \in {\mathbb P}(\H)$ not orthogonal to $|0\>$.
\end{corollary}

\begin{remark}
Concerning the set $X_+(\H)$, we also easily conclude that:
\begin{eqnarray*}
X_+(\H)= {\rm Ran}\, F & = & \partial^i\D_+(\H) \cap (Y(\H) \cup {\rm iint}\, \D(\H))\\
& =& (\D_+(\H)\cap Y(\H)) \cup (\D_{0+}(\H)\cap {\rm iint}\,\D(\H)).
\end{eqnarray*}	
\end{remark}

Note that $\E(\D_+(\H))\not\subset X_+(\H)$ (since $X_+(\H)$  does not contain extreme points of rank less than $N$). This is possible because formula (\ref{C18-eq}) generates only the interior of $\D_+(\H)$, and not the entire set $\D_+(\H)$. Conversely, we also have, in general, $X_+(\H) \not\subset \E(\D_+(\H))$ (in particular $X_+(\H)$ may contain elements from the boundary $\partial^i \D(\H)$ which may have an empty nodal set).

\subsubsection{Simple examples}

In this section we illustrate some of our results in the cases of the $2$-dimensional complex Hilbert spaces $\H=$ span$\{|0\>,|n\>\}$, for $n=1,\dots,4$; and in the case of the $3$-dimensional complex Hilbert space $\H^2=$ span$\{|0\>,|1\>,|2\>\}$. 

In the sequel, we will use the general formula for the Wigner function of $|m\rangle\langle n|$ with $m\le n$; $n,m \in \IN_0$:
\begin{equation}\label{Wignermn}
W_{|m\rangle\langle n|}(x,\xi)
=\frac{1}{\pi}(-1)^m
\sqrt{\frac{m!}{n!}}\;(\sqrt{2}\,z)^{\,n-m}\;
L_{m}^{(n-m)}\bigl(2r^2\bigr)\,e^{-r^2},
\end{equation}
where $L_{m}^{(n-m)}$, is the generalized Laguerre polynomial and $z=r e^{i\theta}=x+i\xi$ (we redefined $z$ as a complex variable to simplify the presentation). In particular, 
$W_{\rho_0}(x,\xi)=W_{\ket0\bra0}(x,\xi)=
\frac{1}{\pi}\,e^{-r^2}$.\\

\paragraph{\bf 2-dimensional Hilbert space.}

Let $\H=$ span$\{|0\>,|n\>\}$. A general density matrix $\rho \in \D(\H)$ is of the form:
\begin{equation}\label{DM}
\rho=\rho(p,s):=(1-p) \,|0\>\<0| +p \, |n\>\<n| + s |0\>\<n| + \ol{s}|n\>\<0|\, ,\quad (p,s) \in R
\end{equation}
where 
\begin{equation}\label{R0}
(p,s)\in R:=\{(p,s) \in \RE\times\CO: 0\le p\le 1 \, \wedge |s|^2\le p(1-p)\}. 
\end{equation}
In this case $Y(\H)$ (cf.~Eq.~(\ref{Y})) is the set of pure states in $\D(\H)$, excluding $\rho_0=|0\>\<0|$,
$$
Y(\H)=\{\rho_1= |u\>\<u|: |u\> \in \H \wedge |u\> \not= |0\> \}.
$$
Let $\H \ni |u\>=a |0\>+b|n\>$ for some $a,b\in \CO$ such that $|a|^2+|b|^2=1$ and $b\not=0$. A generic element in $Y(\H)$ can be written:
$$
\rho_1= (1-p) \,|0\>\<0| +p \, |n\>\<n| + s |0\>\<n| + \ol{s}|n\>\<0|\, , \quad 0< p\le 1 \,, \quad |s|^2=p(1-p)
$$
where $p=|b|^2$ and $s=a\ol{b}$. Hence: 
$$
\rho_1=\rho(p,s) \mbox{ for some } (p,s) \in \partial R \backslash \{(0,0)\}
$$
From Eq.~(\ref{Wignermn}), we easily get (assuming that $s$ is real):
\begin{eqnarray*}
W_{\rho_1}(r,\theta)&=& (1-p)\,W_{\ket{0}\bra{0}}(r,\theta) + p\,W_{\ket{n}\bra{n}}(r,\theta) + s\Bigl[ W_{\ket{0}\bra{n}}(r,\theta) + W_{\ket{n}\bra{0}}(r,\theta)\Bigr]\\
&=& \frac{e^{-r^2}}{\pi} \left[ (1-p) + p\,(-1)^n L_n\Bigl(2r^2\Bigr) + 2s\, \frac{(\sqrt2r)^n\cos(n\theta)}{\sqrt{n!}} \right]
\end{eqnarray*}
Since $\N(\rho_1)\not=\emptyset$, we have from Eq.~(\ref{F2}):
\begin{equation*}
k_0(\rho_1)=\min\limits_{r\ge 0, \, 0\le \theta <2 \pi } \tfrac{W\rho_1(
r,\theta)}{W\rho_0(r,\theta)} <0.
\end{equation*}
Noting that the minimum value with respect to \(\theta\) is reached when
$$
\cos(n\theta) = -\,\operatorname{sgn}(s),
$$
and taking into account that $|s|^2=p(1-p)$, we obtain:
\begin{equation}
k_0(\rho_1)=\min_{r\ge0} \left[(1-p) + p\,(-1)^n L_n\Bigl(2r^2\Bigr) - 2\sqrt{p(1-p)} \frac{(\sqrt2r)^n}{\sqrt{n!}} \right].
\end{equation}
Once we have computed $k_0(\rho_1)$, we can determine the generators of the WPS (cf.~Eq.~(\ref{C18-eq})): 
\begin{equation}\label{rho_+}
\rho_+=F(\rho_1)= (1-t_0(\rho_1)) \rho_0 + t_0(\rho_1)\rho_1
\end{equation}
where $t_0(\rho_1) =\frac{1}{1-k_0(\rho_1)}$. These states form the sets $X_+(\H)$. A straightforward calculation shows that:
\begin{equation}\label{Eqrho_+}   
\rho_1=\rho(p,s) \Longrightarrow \rho_+=F(\rho_1)=\rho(t_0p,t_0s) \quad \mbox{where } t_0=t_0(\rho_1).
\end{equation}
The value of $t_0(\rho_1)$ can be obtained analytically for $n=1,2$. After a simple calculation we get: 
$$
n=1 \,\Longrightarrow \, k_0(\rho_1)= -p  \,\Longrightarrow \, t_0(\rho_1)= \frac{1}{1+p}
$$
$$
n=2 \,\Longrightarrow \, k_0(\rho_1)= -p - 2\sqrt{2p(1-p)}  \,\Longrightarrow \, t_0(\rho_1)= \frac{1}{1+p + 2\sqrt{2p(1-p)}}
$$
In the other cases the value of $t_0(\rho_1)$ can be obtained numerically.

Finally, it follows from \cref{C18} that the interior of $\D_+(\H)$ is  the union of all intervals of the form $(\rho_0,F(\rho_1))$, where $\rho_1\in Y(\H)$. From Eq.~(\ref{Eqrho_+}) we easily conclude that if $\rho_1=\rho(p,s)\in Y(\H)$ then:
$$
\rho\in (\rho_0,F(\rho_1)) \, \Longleftrightarrow \,  \rho=\rho(tp,ts) \, , \,  0<t<t_0(\rho_1) 
$$
and so 
$$
\D_+(\H)=\{\rho{(tp,ts)}\,|\, (p,s) \in \partial R \; \mbox{ and } \; 0\le t \le t_0(\rho{(p,s)}) \},  
$$ 
where $R$ was defined in Eq.~(\ref{R0}).

\begin{figure}[t]
  \centering
  \includegraphics[width=0.50\textwidth]{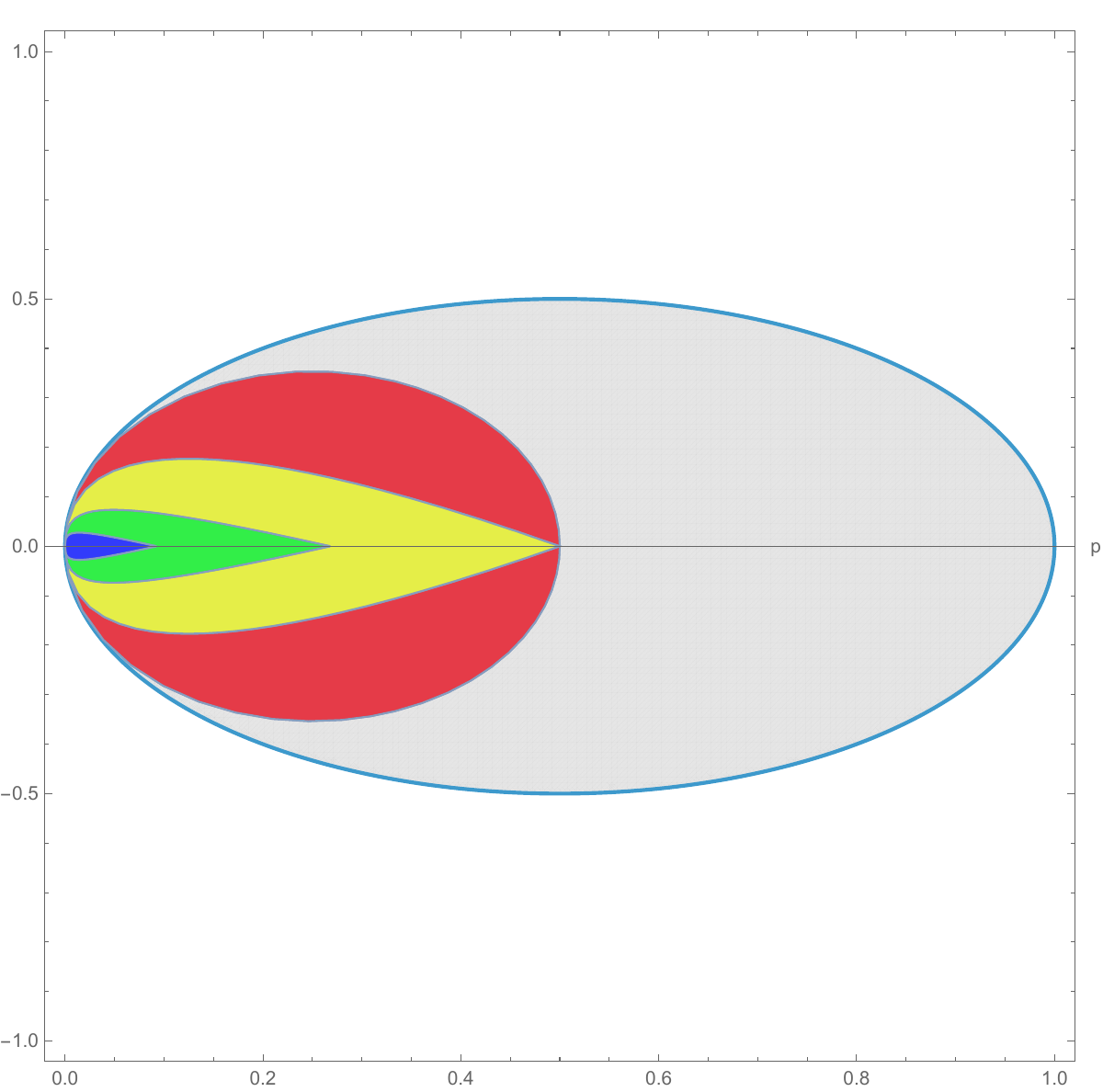}
  \caption{\small Set of density matrices $\D(\H)$ (grey area) and of WPS $\D_+(\H)$ for the case of real $s$. Each point $(p,s)$ inside the larger ellipse  represents a density matrix 
$
\rho{(p,s)}= (1-p) \,\rho_0+ p \,|n\>\<n| +s [|0\>\<n| + |n\>\<0| ].
$
The blue line (without the point $(0,0)$) represents the set $Y(\H)$. The origin is the state $\rho_0=|0\>\<0|$. The coloured domains, in decreasing order of their area size, represent the sets of WPS for the cases of $n=1,\dots,4$, respectively. Their boundaries (without the point $(0,0)$) represent the sets $X_+(\H)$ of the states $\rho_+$.}
  \label{fig:2d}
\end{figure}

\cref{fig:2d} displays the set of density matrices $\D(\H)$, $\H=$ span$\{|0\>,|n\>\}$, and the set of WPS $\D_+(\H)$ for $n=1,\dots,4$. Each point $(p,s)$ inside the larger ellipse represents a density matrix $\rho{(p,s)}$ of the form (\ref{DM}) with real $s$. 
The sets of WPS for $n=1,2,3,4$ are represented by the red, yellow, green and blue regions, respectively. 

The present results for the 2-dimensional Hilbert space are echoed in our companion paper~\cite{physicspaper}. There, we focus on the Hilbert space $\mathcal{H}=\mathrm{span}\lbrace\ket{0},\ket{1}\rbrace$, which we analyze within the framework of quantum optics. We show that a pure-loss channel, i.e. a quantum channel implemented by a beam splitter coupled to a vacuum environment, shares with the map $F$ the ability to reconstruct all states on the boundary of the Wigner-positive set starting from pure states. Consequently, we identify the set $X_+$ with the so-called beam-splitter states, which are especially relevant in the context of Wigner positivity.\\\

\paragraph{\bf 3-dimensional Hilbert space.} 

In the previous example the formulation of $\D_+(\H)$ as a union of affine cones, as presented in \cref{C18-B}, reduces trivially to a union of intervals. This is because, in this case, the sets $\D(\<v\>^\perp)$ are formed by a single pure state from $Y(\H)$, a consequence of dim $\H=2$. In higher dimensions the characterization of the states $\rho\in Y(\H)$ is a more difficult problem, and the partition of $\D_+(\H)$ into (affine) cones can be useful. In this section we illustrate these features for a $3$-dimensional Hilbert space. 

We consider the case $\H=\H^2=$ span$\{|0\>,|1\>,|2\>\}$. 
Let us introduce the notation $Y^N\coloneqq Y(\H^N)$.
It follows from Eq.~(\ref{Yold}) that 
\begin{equation}
Y^2 = \bigcup_{\<v\> \in {\mathbb P}(\H^2) \wedge \<0|v\> \not=0} {\rm iint} \,  \D(\<v\>^\perp) 
\end{equation}
where ${\mathbb P}(\H^2)$ is the projective space of rays in $\H^2$.
 Moreover, every $\rho\in Y^2$ belongs to one, and only one, of the sets $\D(\<v\>^\perp)$. 
We have from Corollary \ref{C18-B}: 
$$
 \D^2_+= \ol{\bigcup_{\<v\> \in {\mathbb P}(\H^2) \wedge \<0|v\> \not=0} {\rm conv}\big(\left\{\rho_0 \right\} \cup F\big(\D(\<v\>^\perp) \big)\big)}.
$$
Let us study the affine cones ${\rm conv}\big(\left\{\rho_0 \right\} \cup F\big(\D(\<v\>^\perp) \big)\big)$ for some vectors $v\in \H^2$. 

We start with the simplest case $v=|0\>$. Then: 
$$
\<v\>^\perp=\<0\>^\perp= \mbox{span}\{|1\>,|2\>\}
$$ 
and a general density matrix $\rho_1\in \D(\<0\>^\perp) $  is of the form:
\begin{equation}\label{rho_1-3dim}
\rho_1= (1-p) \,|1\>\<1| +p \, |2\>\<2| + s |1\>\<2| + \ol{s}|2\>\<1|\, , \quad 0\le p\le 1 \,, \quad |s|^2\le p(1-p).
\end{equation}
From Eq.~(\ref{Wignermn}), we easily find that for real $s$
\begin{equation}\label{W12}
W_{\rho_1}(r,\theta)
=\frac{e^{-r^2}}{\pi}\Bigl[
-(1-p)(1-2r^2)
\;+\;p\,(1-4r^2+2r^4)
\;-\;4s\,r\cos\theta\,(1-r^2)
\Bigr].
\end{equation}
We can then calculate $\rho_+=F(\rho_1)$ using Eqs.~(\ref{F1},\ref{F2}). Since $\rho_1$ is not a WPS, we have:
\begin{eqnarray*}
k_0(p,s)&=&\min_{r\ge0 \, , \, 0\le \theta <2\pi}\left[\tfrac{W_{\rho_1}(r,\theta)}{W_{\rho_0}(r,\theta)} \right]\\
&=&\min_{r\ge0}\Bigl[
-(1-p)(1-2r^2)
+p\,(1-4r^2+2r^4)
-4|s(1-r^2)|r
\Bigr],
\end{eqnarray*}
where we set $\cos\theta=\mathrm{sgn}(s(1-r^2))$, and used $W_{\rho_0}(r,\theta)
=\frac{1}{\pi}\,e^{-r^2}$. Computing $k_0$, and substituting in $t_0(p,s)= \frac{1}{1-k_0(p,s)}$, we can determine
$$
\rho_+(p,s) =F(\rho_1)= \big(1-t_0(p,s)\big)\rho_0 + t_0(p,s) \rho_1\,.
$$

Let us introduce the following notation. For $v\in \H: \<0|v\>\not=0$, let
$$
\D_{v}:={\rm conv}\big(\left\{\rho_0 \right\} \, \cup \,\D(\<v\>^\perp )\big)
$$
be the affine cone of density matrices in $\D(\H)$ whose base is perpendicular to $|v\>\<v|$. 

 \cref{fig:combined1} displays the affine cone of density matrices $\D_{|0\>}= {\rm conv}\big(\left\{\rho_0 \right\} \, \cup \,\D(\<0\>^\perp )\big)$, and the affine cone of WPS $\D_+^2 \,\cap \, \D_{|0\>} ={\rm conv}\big(\left\{\rho_0 \right\} \,\cup \, F\big(\D(\<0\>^\perp)\big)\big)$, for the case of real $s$. Each point $(p,q,s)$ in the larger domain represents a density matrix $\rho(p,q,s) \in \D_{|0\>}$, that can be written in the form:
$$
\rho(p,q,s):= (1-p-q) \,\rho_0+ q \,|1\>\<1| +p \, |2\>\<2| + s \left[|1\>\<2| + |2\>\<1|\right].
$$
Note that if $p+q=1$ then $\rho(p,q,s) \in \D(\<0\>^\perp) \subset Y^2$ is of the form given by Eq.~(\ref{rho_1-3dim}). These states are on the base of the cone. Moreover:
$$
\D_+^2 \cap \D_{|0\>} = \{\rho(tp,tq,ts): \rho_(p,q,s) \in \D(\<0\>^\perp) \wedge \, 0\le t \le t_0(p,s)  \}.
$$

\begin{figure}[t]
  \centering
  \includegraphics[width=0.45\textwidth]{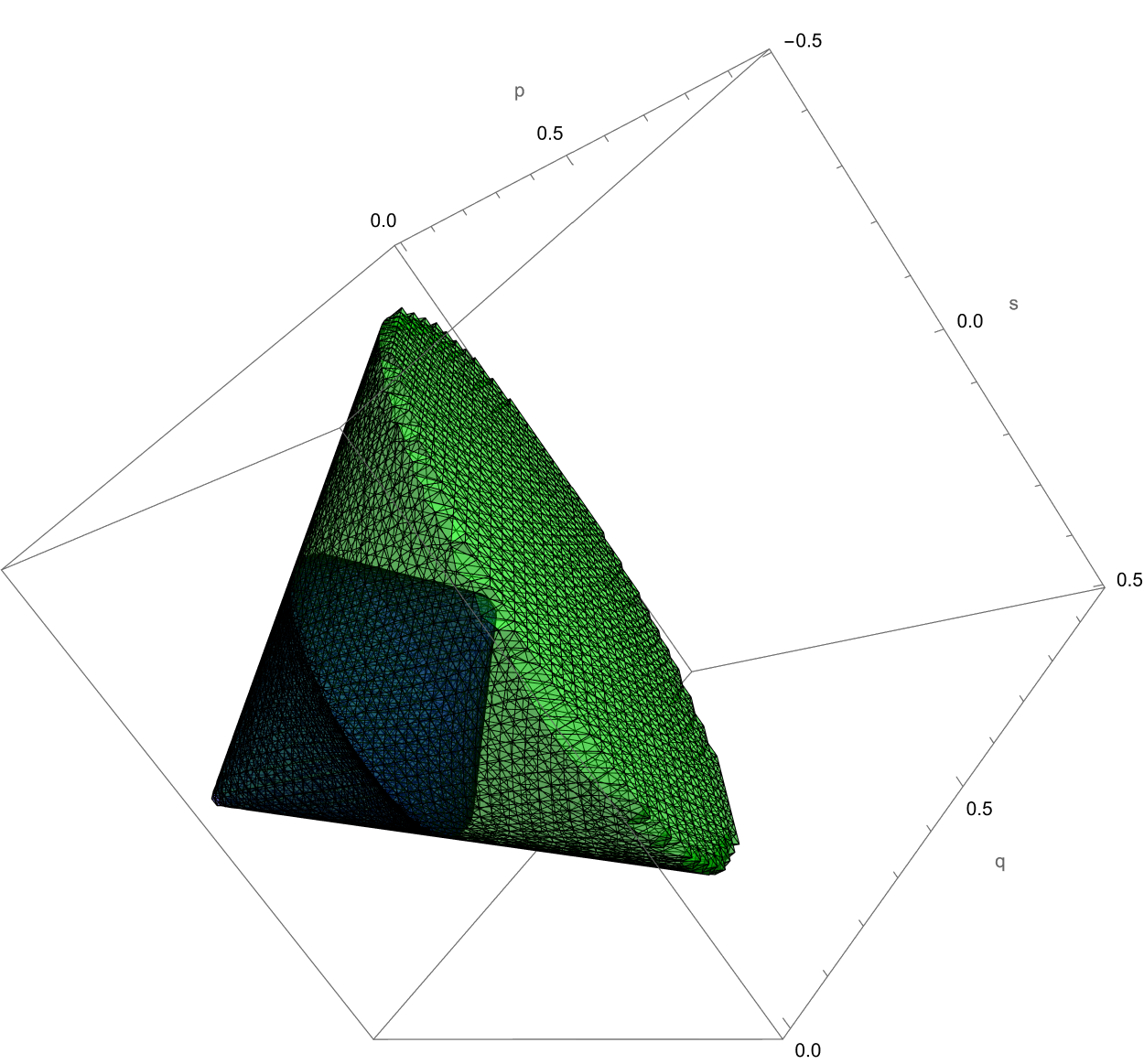}
  \hspace{1cm}
  \includegraphics[width=0.40\textwidth]{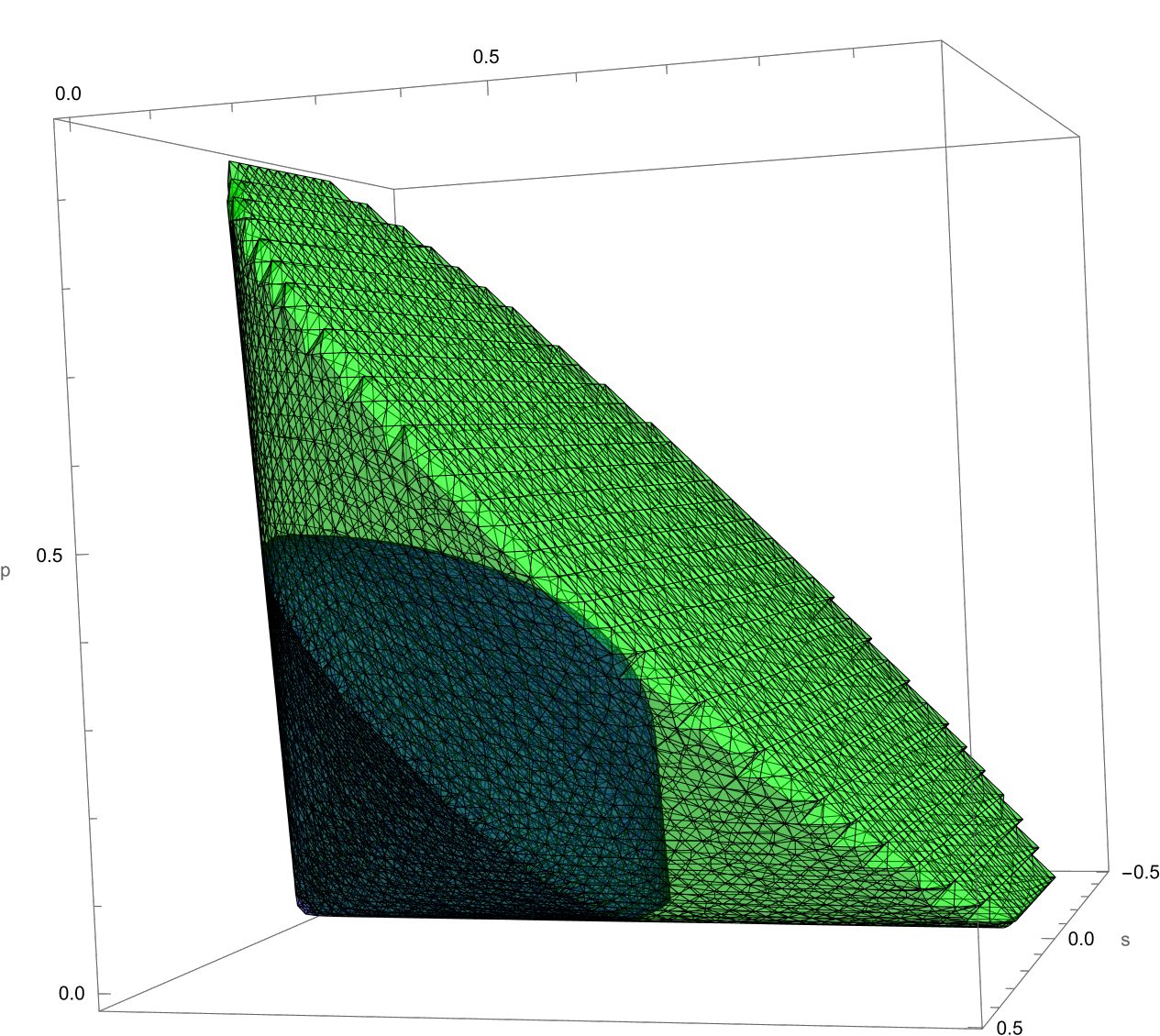}
  \caption{\small Two views of the (affine) cones of density matrices $\D_{|0\>}$ (green domain) and of WPS $\D_+^2 \cap \D_{|0\>}$ (blue domain) for the case of real $s$. Each point $(p,q,s)$ in the larger domain represents a density matrix 
$
\rho(p,q,s)= (1-p-q) \,\rho_0+ q \,|1\>\<1| +p \, |2\>\<2| + s \left(|1\>\<2| + |2\>\<1|\right).
$
The vertex of the cone is the state $\rho_0=|0\>\<0|$, and the base of the cone is the set $\D(\<0\>^\perp) \subset Y^2$ (restricted to real $s$). A cross section of this figure (for s=0) is displayed in \cref{fig:combined} (left and right) for a=0 and c=0, respectively.}
  \label{fig:combined1}
\end{figure}

Let us extend the previous discussion to the case of a general vector $v \in \H^2$ such that $\<0|v\> \not=0$:
$$
v=bd |0\> - ad |1\> - bc |2\> \, , \quad b,d \not=0.
$$
Then 
$$
\<v\>^\perp = \mbox{span} \{u,w\} 
$$
where 
$$
u=\tfrac{1}{\sqrt{a^2+b^2}}\left( a|0\>+b|1\>\right) \quad , \quad w=\tfrac{1}{\sqrt{c^2+d^2}}\left( c|0\>+d|2\>\right), 
$$
and $\rho_1 \in \D(\<v\>^\perp)\subset Y^2$ is of the general form:
$$
\rho_1= (1-p) \,|u\>\<u| +p \, |w\>\<w| + s |u\>\<w| + \ol{s}|w\>\<u|\, , 
$$
where, as usual, $0\le p\le 1 \,, \, |s|^2\le p(1-p)$.

To simplify the presentation, let us restrict to the case $s=0$. The key parameter $k_0(p)$ can then be written in the form:  
$$
k_0(p)={\rm min}\left\{0 \, , \, \tfrac{W_{\rho_1}(r,\theta)}{W_{\rho_0}(r,\theta)}\,\big| \, r\in \RE^+_0\, ,\, 0\le \theta < 2\pi \right\}
$$
where
\begin{equation}\label{R}
\tfrac{W_{\rho_1}(r,\theta)}{W_{\rho_0}(r,\theta)}=(1-p) \tfrac{W_{|u\>\<u|}(r,\theta)}{W_{\rho_0}(r,\theta)} + p \tfrac{W_{|w\>\<w|}(r,\theta)}{W_{\rho_0}(r,\theta)}
\end{equation}
and can be easily calculated after rewriting the last expression in terms of Laguerre polynomials. Let us define:
$$
W_{ij}= \frac{W_{|i\>\<j|}+W_{|j\>\<i|}}{2W_{\rho_0}} \, ,\quad i,j=0,1,2
$$
These terms can be computed using Eq.~(\ref{Wignermn}). We get, after introducing the notation $x=\cos\theta$:
\begin{eqnarray}\label{W_{ij}}
W_{01}(r,x)=	\sqrt{2} rx\,, \quad && \quad W_{02}(r,x)=\sqrt2 r^2(2x^2-1)\nonumber\\
W_{11}(r,x)=-(1-2r^2)\,, \quad && \quad W_{22}(r,x)=1-4r^2+2r^4
\end{eqnarray}
Since:
\begin{eqnarray*}
\frac{W_{|u\>\<u|}(r,x)}{W_{\rho_0}(r,x)}&=&\frac{1}{|a|^2+|b|^2}\left( |a|^2 + 2 \mbox{Re}(a\ol{b}) W_{01} + |b|^2 W_{11}\right)\\
\frac{W_{|w\>\<w|}(r,x)}{W_{\rho_0}(r,x)}&=&\frac{1}{|c|^2+|d|^2}\left( |c|^2 + 2 \mbox{Re}(c\ol{d}) W_{02} + |d|^2 W_{22}\right)
\end{eqnarray*}
after substituting Eq.~(\ref{W_{ij}}) into these expressions, we can calculate Eq.~(\ref{R}) explicitly, and then $k_0(p)$ numerically. Following exactly the same steps as before, we determine $t_0(p)$, $\rho_+$, and finally the set of WPS in $\D_v={\rm conv}\big(\left\{\rho_0 \right\} \, \cup \,\D(\<v\>^\perp )\big)$.

\begin{figure}[t]
  \centering
  \includegraphics[width=0.40\textwidth]{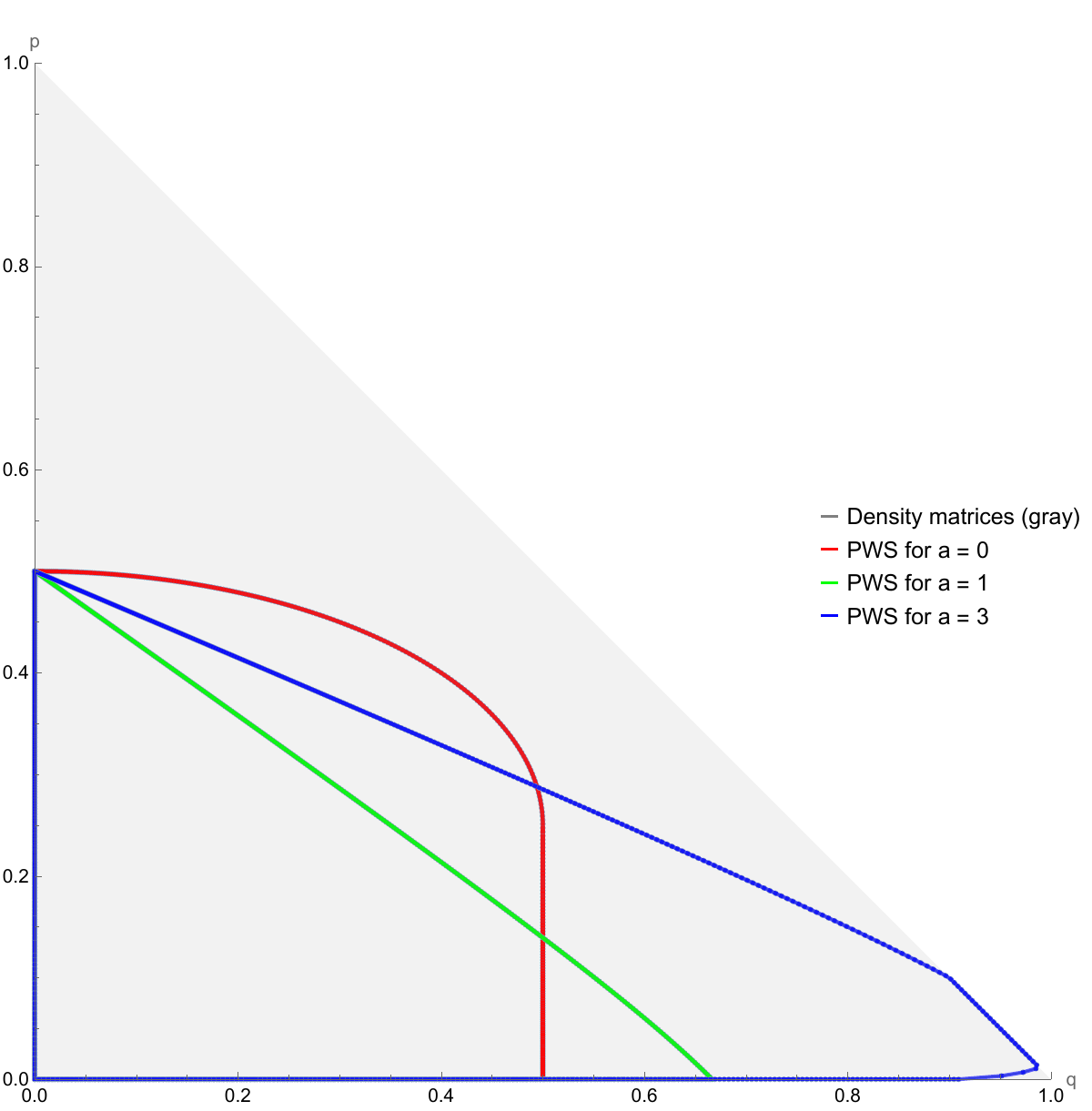}
  \hspace{1cm}
  \includegraphics[width=0.40\textwidth]{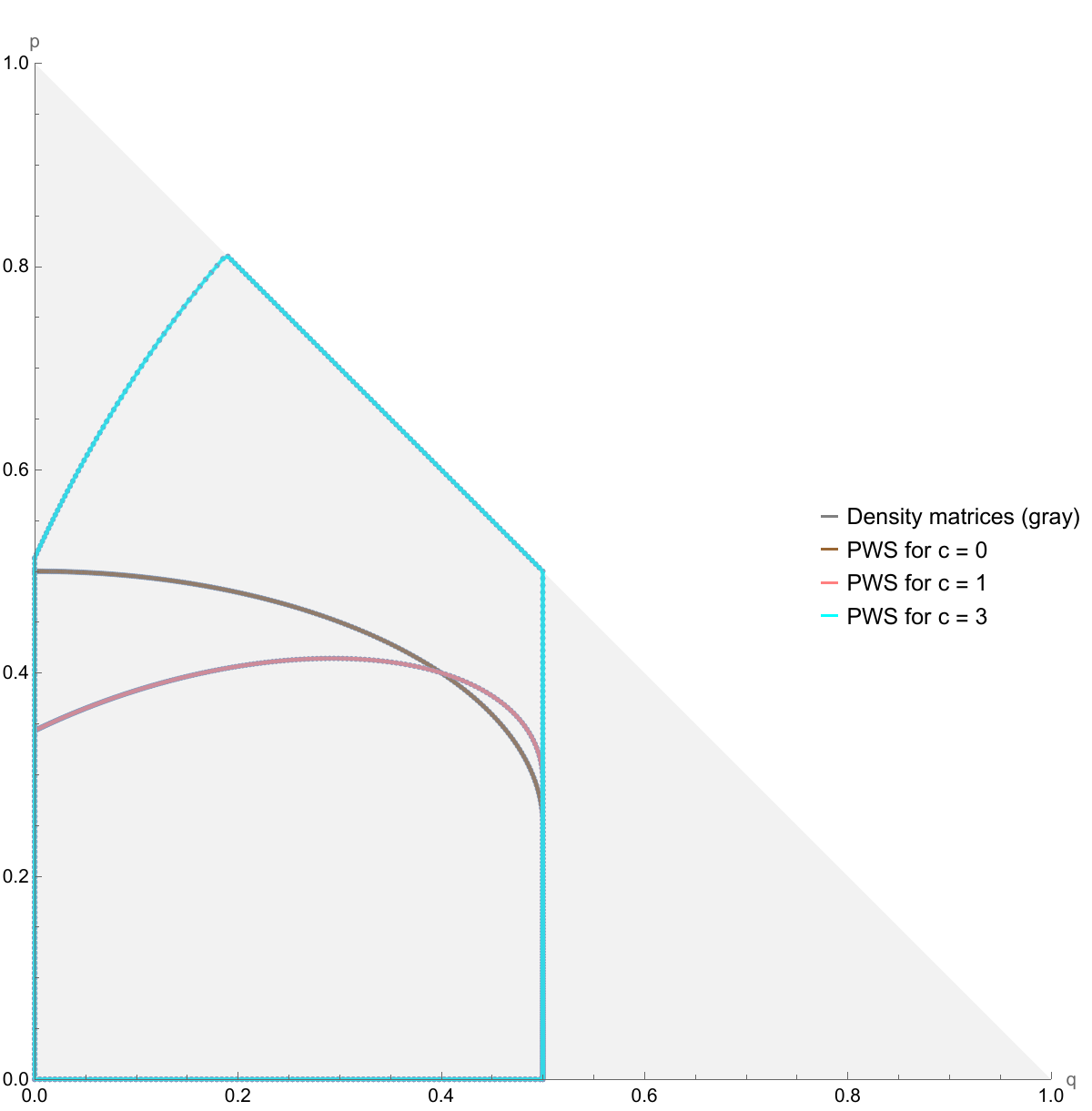}
  \caption{\small Affine cones of density matrices $\D_v$ (grey domain) and of WPS $\D_+^2 \cap \D_v$ for the case $s=0$. Each point $(p,q)$ in the triangular domain represents a density matrix 
$
\rho{(p,q)}= (1-p-q) \,\rho_0+ q \,|u\>\<u| +p \, |w\>\<w|
$
for some $u,w\in \<v\>^\perp$. 
The origin is the state $\rho_0=|0\>\<0|$. The hypotenuse is the set $\D(\<v\>^\perp)$ (restricted to $s=0$). The points inside the domains bounded by the coloured lines are the WPS. The figure on the left corresponds to the cases $v=|0\>-a|1\>, \,u=\tfrac{a|0\>+|1\>}{\sqrt{1+a^2}}$ and $w=|2\> $, for $a=0,1,3$. On the right, we have $v=|0\>-c|2\>,\, u=|1\>$ and $w=\tfrac{c|0\>+|2\>}{\sqrt{1+c^2}}$ for $c=0,1,3$.}
  \label{fig:combined}
\end{figure}

\cref{fig:combined} displays the affine cones of density matrices $\D_{v}$, and of WPS $\D_+^2\,\cap \, \D_{v}={\rm conv}\big(\left\{\rho_0 \right\} \,\cup \, F\big(\D(\<v\>^\perp)\big)\big)$, for the case $s=0$, and several different vectors $v$ (parametrized by the values of $(a,b,c,d)$). Each point $(p,q)$ in the triangular domain represents a density matrix $\rho(p,q) \in \D_{v}$:
$$
\rho(p,q)= (1-p-q) \,\rho_0+ q \,|u\>\<u| +p \, |w\>\<w| \, , \quad 0\le q,p \le 1.
$$
The points $(p,q)$ such that $p,q\ge 0$ and $p+q=1$ represent states $\rho(p,q) \in \D(\<v\>^\perp)$.
Moreover:
$$
\D_+^2 \cap \D_{v} = \{\rho(tp,tq):\, p,q\ge 0\, \wedge \,  p+q=1 \wedge \, 0\le t \le t_0(p)  \}.
$$
is the set of WPS in the cone $\D_v$. 

Here again, our results for the 3-dimensional Hilbert space are closely connected to those presented in our companion work~\cite{physicspaper}. There, we consider the family of states $\rho = (1-p-q)\ket{0}\bra{0} + q\ket{1}\bra{1} + p\ket{2}\bra{2}$ and focus on their Wigner-positive restriction (delimited by the red curve in Fig.~\ref{fig:combined}-Left, or equivalently by the brown curve in Fig.~\ref{fig:combined}-Right). We find that the curved boundary of that set can be obtained as the evolution of a specific state -- a beam-splitter state -- through a new quantum map, which we introduce as the \emph{Vertigo map}. As highlighted in that companion paper, the Vertigo map displays remarkable properties with respect to extreme Wigner-positive states.

\subsection{Geometry of Wigner-positive states in infinite dimension}
\label{sec:geom_infinite}

We now turn to the infinite-dimensional case.
Let $V$ be a Hausdorff, locally convex TVS and let $S\subset V$ be convex and compact. The Krein--Milman theorem \cite{Krein} is one of the central results of convex analysis. It states that
\begin{equation}\label{KM}
S=\overline{\mbox{\rm conv}}(\E(S))
\end{equation}
where, for an arbitrary set $A\subseteq V$, the closed convex hull of $A$, denoted $\overline{\mbox{\rm conv}}(A)$, is the smallest closed, convex set in $V$ that contains $A$ (cf.~Definition \ref{Aff+Conv}). It can be shown that if $V$ is finite-dimensional and $A$ is compact, then $\overline{\mbox{\rm conv}}(A)={\rm conv}(A)$. In infinite dimensions, $\overline{\mbox{\rm conv}}(A)$ coincides with the closure of the set of all {\it finite} convex combinations of the elements of $A$, i.e.
$$
\overline{\mbox{\rm conv}}(A)= \overline{{\rm conv}(A)}
$$   
The closure of a set $A \subset V$ of course depends on the topology in $V$. 

The topology is {\it metrisable} if and only if it can be induced by a metric. For a locally convex TVS this means that the topology is determined by a countable family of seminorms (cf.~\cite[Proposition 2.1, Chap. IV]{Conway}). 
All metrisable spaces $V$ (and in particular all metric spaces, and hence all norm spaces) are first countable which in turn are Fr\'echet--Urysohn spaces (sometimes just called Fr\'echet spaces) \cite[Theorem 1.6.14]{Engelking}. The latter are defined by the property that every subset $A\subset V$ satisfies $\overline{A}=$sc$(A)$ where sc$(A)$ denotes the sequential closure of $A$:
$$
{\rm sc}\,(A)\coloneqq\{x\in V:\, \mbox{there is a sequence } (a_n)_{n\in \IN} \subset A \mbox{ with } a_n \to x \}~.
$$
Therefore, if the topology is metrisable, Eq.~(\ref{KM}) can be rewritten in the form:
\begin{equation}\label{KM2}
S={\rm sc}\big({\mbox{\rm conv}}\big(\E(S)\big)\big)
\end{equation}

Note that the compactness of $S$ also depends on the topology. One of the limitations of the Krein--Milman theorem is that if $V$ is an infinite-dimensional Banach space then the unit ball is not compact in the norm topology.

\subsubsection{A Krein--Milman theorem for \texorpdfstring{$\D_+$}{}}

For $\H=L^2(\RE)$, the set  
of WPS $\D_+$ is not compact with respect to the trace norm (cf.~\cref{Compact_infinite_dim}), and so the Krein--Milman theorem does not apply directly. We will prove, nevertheless, that the {\it Krein--Milman relation} given by Eq.~(\ref{KM}) still holds for $\D_+$. Note that the conditions in the Krein--Milman theorem (namely, compactness) are sufficient but not necessary conditions. 

Our approach requires that $\B_1=\B_1(\H)$ be endowed with the weak* topology (denoted w* topology). Recall that the space of trace-class operators $\B_1$ is the topological dual of the set of compact operators $\K=\K(\H)$ with the standard operator norm, i.e.~$\B_1 =\K^*$ \cite[Theorem VI.26]{ReedSimon}. The w* topology in $\B_1$ is defined by the family of seminorms ${\mathcal F}=\{p_k:\, K \in \K\}$, where:
$$
p_K:\B_1 \to [0,\infty) \, , p_K(a) = |{\rm tr}(a K)| 
$$
This topology gives the following convergence of sequences:
$$
\B_1 \ni a_n \xrightarrow[n\to \infty]{{\mathrm w}^*} \, a \in \B_1 \quad \Longleftrightarrow \quad \mbox{tr}\,(a_n K) \xrightarrow[n\to \infty]{} \mbox{tr}\,(a K)\, , \quad  \forall \, K\in \K
$$
Let us collect a few basic properties that will be used to prove the main theorem:
\begin{enumerate}
\smallskip
\item[(A)] $\B_1$ with the w* topology is a Hausdorff, locally convex TVS. This is a general property of the dual (endowed with the w* topology) of every normed space \cite[Definition 1.2 and Example 1.8, Chap. IV]{Conway}. Note that in Conway's book \cite{Conway} the definition of locally convex TVS already includes the Hausdorff property. 
\smallskip
\item[(B)] Since $L^2$ is separable, $\K$ is also separable, and moreover is a Banach space. It also follows from a general result \cite[Corollary 2.6.20]{Megginson} that the (relative) w* topology on a bounded subset of $\B_1$ is metrisable (but not on the whole space $\B_1$). Note that a set $X \subset \B_1$ is norm-bounded if and only if it is w*-bounded \cite[Theorem 2.6.7]{Megginson}. $\B_1=\K^*$ is equipped with the dual (operator) norm, which coincides with the trace norm:
$$
\| a \|_*= {\rm sup} \{|{\rm tr} (a K) | : K\in \K \wedge \|K\| \le 1 \} =\|a\|_1
$$
\smallskip
\item[(C)] The unit ball $B_1=\{a \in \B_1 : \| a \|_1 \le 1 \}$ is w*-compact. This property follows from the Banach--Alaoglu theorem  \cite[Theorem IV.21]{ReedSimon}. Since $B_1$ is bounded, the w* topology on $B_1$ is metrisable. We can then conclude that $B_1$ with the relative w* topology is a Fr\'echet--Urysohn subspace \cite[Theorem 1.6.14]{Engelking} and so every subset $A\subset B_1$ satisfies $\overline{A}=$sc$(A)$, where both the closure and the sequential closure are with respect to the w* topology.
\smallskip
\item[(D)] The cone Con ${\D_+}=\{\lambda \rho \, | \, \rho \in \D_+ \, \wedge \, \lambda \in \RE^+_0 \}$ is w*-closed \cite[Theorem 7, Corollary 31]{Lami}.
\smallskip
\item[(E)] Let $\rho_n,\rho \in \D(\H)$. Then $\rho_n \xrightarrow[n\to \infty]{{\mathrm w}^*} \, \rho$ if and only if $ \rho_n \xrightarrow[n\to \infty]{\|\, \,\|_1} \rho$ \cite[Lemma 4.3]{Davies}.

\end{enumerate}

\smallskip

We can now prove that:

\begin{theorem}\label{KMR}
Let $\H=L^2(\RE)$ and let $\D_+=\D_+(\H)$. Then
$$
\D_+=\overline{\mbox{\rm conv}}\,(\E(\D_+)) 
$$
where the closure is with respect to the trace norm.
\end{theorem}

\begin{proof}
Let us define 
$$
\widetilde\D_+ \coloneqq B_1 \, \cap \, {\rm Con}\,\D_+ = \{\lambda \rho \, | \, \rho \in \D_+ \, \wedge \, \lambda \in [0,1] \}, 
$$ 
which is a subset of $\B_1$. In view of (C) and (D) above, this set is w*-compact (it is a w*-closed subset of a w*-compact set). 

Given (A) above, and since $\widetilde\D_+$ is also convex, the Krein--Milman theorem holds:
\begin{equation}\label{cchw}
\widetilde\D_+=\overline{\mbox{\rm conv}}^{{\mathrm w}^*}\, (\E(\widetilde\D_+))  
\end{equation}
where $\overline{\mbox{\rm conv}}^{{\mathrm w}^*}(A)$ denotes the w*-closure of ${\rm conv}$$(A)$.

Consider an arbitrary element $\rho \in \D_+ \subset \widetilde\D_+$. Taking into account (C), 
it follows from (\ref{cchw}) that $\rho \in $ sc$^{{\mathrm w}^*}($${\rm conv}$$(\E(\widetilde\D_+))$, and so there exists a sequence $(\widetilde\rho_n)_{n\in \IN} \subset$ ${\rm conv}$$\,(\E(\widetilde\D_+))$ such that
\begin{equation} \label{wtrho}
\wt\rho_n \xrightarrow[n\to \infty]{{\mathrm w}^*} \rho.
\end{equation}
Since $\E(\wt\D_+)= \E(\D_+) \cup \{\mathbf{0}\}$, where $\mathbf{0}$ is the null operator, we can write  each $\wt\rho_n$ in (\ref{wtrho}) in the form 
$$ 
\wt\rho_n = \lambda_0 \, \mathbf{0} + \sum_{i=1}^m \lambda_i e_i= \sum_{i=1}^m \lambda_i e_i \quad , \quad 0 \le \lambda_i \le 1 \, ,\quad 0< \sum_{i=1}^m \lambda_i \le 1
$$
where $e_i \in \E(\D_+)$, and $m,\lambda_i,\, i=1,\dots,m$ depend of $n$. 

Let us define $\rho_n=\wt\rho_n / k_n$ where $k_n=\|\wt\rho_n \|_1=\sum_{i=1}^m \lambda_i$. Hence, $\rho_n \in$ ${\rm conv}$$\,(\E(\D_+))$. We then have from (\ref{wtrho}):
$$
k_n\rho_n \xrightarrow[n\to \infty]{{\mathrm w}^*} \rho \,\Longleftrightarrow k_n \, {\rm tr}\,(\rho_n C) \xrightarrow[n\to \infty]{} {\rm tr}\,(\rho C)\, , \,\forall C\in \K
$$ 
Since $(k_n)_{n\in \IN}$ is a bounded sequence of real numbers, it has a convergent subsequence $(k_{n_i})_{i\in \IN}$, where $n_i:\IN\to \IN$ is strictly increasing. Hence, $k_{n_i} \to k$, for some $0\le k\le 1$. First consider that $k=0$. Then (using the fact that $|{\rm tr}\,(\rho_{n_i} C)| \le \| C \|$, for all $C\in \K$ )
$$
k_{n_i}  {\rm tr}\,(\rho_{n_i} C) \xrightarrow[i\to \infty]{} 0\, , \, \forall C\in \K \; \Longrightarrow \; \wt\rho_{n_i} \xrightarrow[i\to \infty]{{\mathrm w}^*} 0
$$
which contradicts (\ref{wtrho}). Next consider the case $0<k<1$. We have:
$$
{\rm tr}\,(\rho_{n_i} C) \xrightarrow[i\to \infty]{} \tfrac{1}{k} {\rm tr}\,(\rho C)\, , \, \forall C\in \K \; \Longleftrightarrow \; \rho_{n_i} \xrightarrow[i\to \infty]{{\mathrm w}^*} \tfrac{1}{k}\rho
$$ 
Note that $(\rho_{n_i})_{i\in \IN} \subset \wt\D_+$ but $ \tfrac{1}{k}\rho \notin \wt\D_+$ (since $\| \tfrac{1}{k}\rho \|_1 >1$). This is not possible since $\wt\D_+$ is w*-compact (and thus w*-sequentially closed).  

Therefore, $k_n \to 1$ (otherwise $(k_n)_{n\in \IN}$ would always have a subsequence converging to some $k<1$). We conclude from (\ref{wtrho}) that for all $\rho \in \D_+$ there exists a sequence $(\rho_n)_{n\in \IN} \subset {\rm conv}(\E(\D_+))$, given by $\rho_n=\wt\rho_n/k_n$, and such that
\begin{equation} \label{rhon}
\rho_n \xrightarrow[n\to \infty]{{\mathrm w}^*} \rho \quad \Longleftrightarrow \quad \rho_n \xrightarrow[n\to \infty]{\| \, \|_1} \rho 
\end{equation}
where, in the last step, we have used the relation (E) above. Hence 
$$
\D_+ \subseteq \overline{\mbox{\rm conv}}^{\| \, \|_1}\, (\E(\D_+)) ~.
$$
Since $\D_+$ is closed for the trace norm (cf.~\cref{Compact}), we also have $\overline{\mbox{\rm conv}}^{\| \, \|_1}\, (\E(\D_+)) \subseteq \D_+$, and so the two sets are equal, which concludes the proof.
\end{proof}

\subsubsection{From finite to infinite-dimensional extreme points}
\label{sec:geom_finite_to_infinite}

The next question is how to determine and construct the extreme points (or other sets of generators) of $\D_+$. One partial answer is that we may find extreme points of $\D_+$ by looking for extreme points of its various subsets $\D_+(\H)$, for arbitrary finite-dimensional Hilbert subspaces $\H\subset L^2(\RE)$. We follow this strategy in a companion paper \cite{physicspaper} to generate a large class of explicit extreme points of $\D_+$. The validity of this strategy follows from the next result: 

\begin{theorem}\label{Faces+Extreme}
Let $\H\subset L^2(\RE)$ be finite-dimensional and let $\Pi_\H:L^2(\RE) \to \H$ be the orthogonal projector onto $\H$. Then:
\begin{enumerate}
\item[{\bf(A)}] Let $\rho \in \D_+$. Then $\rho \in \D_+(\H)$ if and only if {\rm tr}$\,(\rho \,\Pi_\H )=1$
\smallskip
\item[{\bf(B)}] If nonempty, $\D_+(\H)$ is a face of $\D_+$
\smallskip
\item[{\bf(C)}] $
\E(\D_+(\H))= \E(\D_+) \cap \D_+(\H) 
$		
\end{enumerate}
\end{theorem}

\begin{proof}

{\bf (A)} Let $\{\psi_k, k \in S\}$ be an orthonormal basis of $\H$. Since $\H$ is finite-dimensional, $S$ is a finite index set. Then $\Pi_\H=\sum_{k\in S} |\psi_k\>\<\psi_k|$ and if $\rho \in \D_+(\H) \subset \D(\H)$, we get:
$$
{\rm tr}\,(\rho \,\Pi_\H ) = \sum_{k\in S} \<\psi_k|\rho \psi_k\> = {\rm tr}\,(\rho )=1.
$$
Conversely, if $\rho \in \D_+ \backslash \D_+(\H)$ then $\rho \notin \D(\H)$ and so:
$$
{\rm tr}\,(\rho \,\Pi_\H ) = \sum_{k\in S} \<\psi_k|\rho \psi_k\> < {\rm tr}\,(\rho )=1.
$$

{\bf (B)} Consider the linear, non-constant functional:
$$
f_\H: \D_+ \longrightarrow \RE; \, \rho \longrightarrow f_\H (\rho)={\rm tr}\,(\rho \,\Pi_\H )
$$
Since $f_\H$ is defined on the convex set $\D_+$ and has the maximum ${\rm max} \{f_\H(\rho)|\, \rho \in \D_+\}=1$, it follows from \cite[Theorem 8.3]{Simon2011} that (cf.~Eq.~(\ref{Faces})):
$$
\F_{\H}= \{\rho \in \D_+: f_\H(\rho)=1 \}
$$
is a proper face of $\D_+$. Moreover, from 1. above we obtain $\F_{\H}=\D_+(\H)$. 

\medskip

{\bf (C)} It is an immediate consequence of  $\D_+(\H) \subset \D_+$, that if $\rho \in \E(\D_+)\cap \D_+(\H)$ then $\rho$ is an extreme point of $\D_+(\H)$. Conversely, $\D_+(\H)$ is a face of $\D_+$ and so $\E(\D_+(\H)) \subseteq \E(\D_+)$ (cf.~Eq.~(\ref{ExtFaces})).
\end{proof}

One should note that the previous theorem is not only valid for $\H=\H^N$, but more generally for all Hilbert spaces $\H=H_{z_0,S}^N$, $z_0\in \RE$, $S\in $ Sp$(2)$ and $N\in \IN_0$. Hence:  

	
$$
\bigcup_{N\in \IN_0, z_0\in \RE^2,S\in {\rm Sp}(2)}\, \E(\D_+(\H^{N}_{z_0,S})) \subseteq \E(\D_+)
$$  

Interestingly, it might be possible to find sets of generators that are proper subsets of $\E(\D_+)$. This is a consequence of the following simple theorem: 

\begin{lemma}\label{chdensity}
Let $V$ be a Fr\'echet--Uryshon TVS (e.g.~a metrisable TVS), and let $A$ and $B$ be subsets of $V$ where $A \subset B$ is a dense subset of $B$ (i.e.~$B \subseteq \overline{A}$). Define
$$
A_1 = {\rm conv}(A), \quad B_1 = {\rm conv}(B),
$$
and
$$
A_2 = \overline{A_1}=\overline{\mbox{\rm conv}}(A), \quad B_2 = \overline{B_1} = \overline{\mbox{\rm conv}}(B).
$$
Then:
\begin{enumerate}
    \item[{\bf(A)}] $A_1$ is dense in $B_1$, i.e.~$B_1 \subseteq \overline{A_1}$.
    \smallskip
    \item[{\bf(B)}] $A_2 = B_2$.
\end{enumerate}
\end{lemma}

\begin{proof}
Since $A \subset B$, it follows from the monotonicity of the convex hull that
$$
A_1 = {\rm conv}(A) \subset {\rm conv}(B) = B_1.
$$

\medskip

\textbf{(A)}
Let $x \in B_1$. Then there exist points $b_1, \dots, b_n \in B$ and non-negative coefficients $\lambda_1, \dots, \lambda_n$ with $\sum_{i=1}^n \lambda_i = 1$, such that
$$
x = \sum_{i=1}^n \lambda_i b_i.
$$
Since $V$ is a Fr\'echet--Uryshon space, the closure and the sequential closure of its subsets coincide. Since $A$ is dense in $B$, for every $b_i$, $i = 1, \dots, n$ there exists a sequence $(a_{i,k})_{k\in\IN} \subset A$ satisfying
$$
a_{i,k} \longrightarrow b_i \quad \text{as } k \longrightarrow \infty.
$$
Define the sequence $(x_k)_{k\in \IN}$ in $A_1$ by
$$
x_k = \sum_{i=1}^n \lambda_i a_{i,k} ~.
$$
Since $a_{i,k} \in A$ for all $i$ and $k$, every $x_k$ is in ${\rm conv}(A) = A_1$. By the continuity of addition and scalar multiplication in $V$ (since $V$ is a TVS), we have
$$
x_k \longrightarrow \sum_{i=1}^n \lambda_i b_i = x ~.
$$
Thus, $x\in \overline{A_1}$, and so
$
B_1 \subseteq \overline{A_1}
$.

\medskip

\textbf{(B)}
By definition, $A_2 = \overline{A_1}$ and $B_2 = \overline{B_1}$. Since $A_1 \subset B_1$, it immediately follows that
$
A_2 \subset B_2.
$
On the other hand, we already know that $B_1 \subseteq \overline{A_1}$. Taking closures on both sides (and noting that closures are closed sets), we get
$$
\overline{B_1} \subset \overline{\overline{A_1}} = \overline{A_1},
$$
so that
$
B_2 \subset A_2 .
$
Combining the two inclusions, we get
$
A_2 = B_2
$
completing the proof.
\end{proof}

\medskip

It follows that: 

\begin{corollary}\label{DenseX}
Let $\mathcal K$ be a dense subset of $\E(\D_+)$ (in the trace norm topology). Then:
$$
\D_+=\overline{\mbox{\rm conv}}(\mathcal K)
$$
\end{corollary}

The preceding discussion motivates the following conjecture, with which we conclude the paper:

\begin{conjecture}

For $\H=L^2(\RE)$, we have:
$$
\D_+= \overline{\mbox{\rm conv}}\left(\bigcup_{N\in \IN}\, \E(\D_+^{N})\right)
$$
i.e.~every WPS is the limit (in the trace norm) of a sequence of finite rank WPS from $\D_+^N$.  
  	
\end{conjecture}

\section*{Acknowledgements}

N.J.C. acknowledges support from the F.R.S.-FNRS under project
CHEQS within the Excellence of Science (EOS) program.

U.C.\ , Z.V.H.\ and J.D.\ acknowledge funding from the European Union's Horizon Europe Framework Programme (EIC Pathfinder Challenge project Veriqub) under Grant Agreement No.~101114899.

N.C.D.\ and J.N.P.\ acknowledge funding from the Fundação para a Ciência e Tecnologia (Portugal) via the research centre GFM, reference UID/00208/2025.

******************************************************************************\\

{\bf Author's addresses:}

\begin{itemize}

\item {\bf Nicolas J. Cerf:}	Centre for Quantum Information and Communication, \'{E}cole polytechnique de Bruxelles,  CP 165, Universit\'{e} libre de Bruxelles, 1050 Brussels, Belgium.

e-mail address: nicolas.cerf@ulb.be\\

\item {\bf Ulysse Chabaud:} DIENS, \'Ecole Normale Sup\'erieure, PSL University, CNRS, INRIA, 45 rue d'Ulm, Paris 75005, France. 

e-mail address: ulysse.chabaud@inria.fr\\

\item {\bf Jack Davis:} DIENS, \'Ecole Normale Sup\'erieure, PSL University, CNRS, INRIA, 45 rue d'Ulm, Paris 75005, France.

e-mail address: jack.davis@inria.fr\\

\item {\bf Nuno C. Dias:} Grupo de Fisica-Matem\'atica, Departamento de Matem\'atica, Instituto Superior
Técnico, Av. Rovisco Pais, 1049-001 Lisbon, Portugal. 

e-mail address: nunocdias1@gmail.com\\

\item {\bf Jo\~ao N. Prata:} Grupo de Fisica-Matem\'atica, Departamento de Matem\'atica, Instituto Superior
Técnico, Av. Rovisco Pais, 1049-001 Lisbon, Portugal,
and
Departamento de Matem\'atica, ISCTE - Instituto Universit\'ario de Lisboa, Av. das
For\c{c}as Armadas, 1649-026, Lisbon, Portugal, 

e-mail address: joao.prata@mail.telepac.pt\\

\item {\bf Zacharie Van Herstraeten:} DIENS, \'Ecole Normale Sup\'erieure, PSL University, CNRS, INRIA, 45 rue d'Ulm, Paris 75005, France.

e-mail address: zacharie.van-herstraeten@inria.fr\\

\end{itemize}

\end{document}